%% file: pphi.tex
\documentclass[nofootinbib]{article}

\input{sections/preamble}

\begin{document}


\input{sections/title_abstract}


\input{sections/intro}

\input{sections/def_models}

\input{sections/amplitudes}

\input{sections/connectedness}

\input{sections/multiscale}

\input{sections/u1_model}

\input{sections/finiteness}

\input{sections/conclusion}


\input{sections/akno}


\appendix

\newpage

\input{sections/appendices/phi_6}

\input{sections/biblio}

\end{document}

%% file: sections/preamble.tex
\usepackage{hyperref}
\usepackage{amsfonts,amssymb,amsmath,exscale,bbm}
\usepackage{footnpag} 
\usepackage[all,knot]{xy}  
\usepackage{color}
\usepackage{graphicx}
\usepackage{epsfig}
\usepackage{subfig}
\usepackage{enumerate}
\usepackage{stmaryrd}

\usepackage{amsmath}
\usepackage{amsfonts}
\usepackage{graphicx}
\usepackage{psfrag}
\usepackage{array}
\usepackage{bbm}
\usepackage{hyperref}
\usepackage{amsthm}
\theoremstyle{definition}
\newtheorem{definition}{Definition}
\newtheorem{theorem}{Theorem}

\newtheorem{proposition}{Proposition}

\newcommand{\cA}{{\mathcal A}}
\newcommand{\cB}{{\mathcal B}}

\newcommand{\cF}{{\mathcal F}}
\newcommand{\cG}{{\mathcal G}}
\newcommand{\cH}{{\mathcal H}}

\newcommand{\cM}{{\mathcal M}}

\newcommand{\cR}{{\mathcal R}}
\newcommand{\cS}{{\mathcal S}}
\newcommand{\cT}{{\mathcal T}}

\newcommand{\cZ}{{\mathcal Z}}

\newcommand{\Inv}{{\mathrm{Inv}}}

\def\inv{{\mbox{\tiny -1}}}

\newcommand\beq{\begin{equation}}
\newcommand\eeq{\end{equation}}
\newcommand{\be}{\begin{equation}}
\newcommand{\ee}{\end{equation}}
\newcommand{\bes}{\begin{eqnarray}}
\newcommand{\ees}{\end{eqnarray}}

\def\vphi{{\varphi}}
\def\vphib{\overline{{\varphi}}}

\newcommand{\one}{\mbox{$1 \hspace{-1.0mm}  {\bf l}$}}

      \def\nn{{\nonumber}}

\newcommand{\U}{\mathrm{U}}

\def\extd{\mathrm {d}}
\newcommand{\e}{\epsilon}

\newcommand\acts\triangleright

\newcounter{letter} \newcounter{numeral} \newcounter{Numeral}

\def\vphi{\varphi}

\def\e{\mbox{e}}
\def\extd{\mathrm {d}}

\newtheorem{theo}{Theorem}
\newtheorem{lemma}[theo]{Lemma}

%% file: sections/title_abstract.tex
\begin{titlepage}
\begin{flushright}
Lpt-Orsay-12-89, AEI-2012-079\\
\end{flushright}

\vspace{20pt}

\begin{center}

{\Large\bf Renormalization of Tensorial Group Field Theories:} \\
\medskip
{\large \bf Abelian $U(1)$ Models in Four Dimensions}
\vspace{15pt}

{\large Sylvain Carrozza$^{a,b}$, Daniele Oriti$^{b}$ and Vincent Rivasseau$^{a} $}

\vspace{15pt}

$^{a}${\sl Laboratoire de Physique Th\'{e}orique, CNRS UMR 8627,\\
 Universit\'{e} Paris XI, F-91405 Orsay Cedex, France, EU\\
}
\vspace{5pt}

$^{b}${\sl Max Planck Institute for Gravitational Physics,\\
Albert Einstein Institute, Am M\"uhlenberg 1, 14476 Golm, Germany, EU\\
}
\vspace{5pt}

{\sl Emails:   sylvain.carrozza@aei.mpg.de,  daniele.oriti@aei.mpg.de,  rivass@th.u-psud.fr
}

\vspace{10pt}

\begin{abstract}
We tackle the issue of renormalizability for Tensorial Group Field
Theories (TGFT) including gauge invariance conditions, with the rigorous tool
of multi-scale analysis, to prepare the ground for applications
to quantum gravity models.  In the process, we define the appropriate
generalization of some key QFT notions, including:
connectedness, locality and contraction of (high) subgraphs. We also
define a new notion of Wick ordering, corresponding to the subtraction
of (maximal) melonic tadpoles. We then consider the simplest examples
of dynamical 4-dimensional TGFT with gauge invariance conditions
for the Abelian $U(1)$ case. We prove that they are
super-renormalizable for any polynomial interaction.
\end{abstract}

\end{center}

\noindent  Pacs numbers:  11.10.Gh, 04.60.-m
\\
\noindent  Key words: Renormalization, group field theory, tensor models,
quantum gravity, lattice gauge theory. 

\setcounter{footnote}{0}

\end{titlepage}

%% file: sections/intro.tex
\section*{Introduction}

A complete theory of quantum gravity and spacetime should be background independent. It should not assume a priori any geometric background for the definition of its fundamental degrees of freedom or dynamical equations. This has been a guiding principle in canonical Loop Quantum Gravity \cite{LQG} and simplicial quantum gravity  \cite{DT, qRC}, but is also a necessary feature of any more fundamental formulation of string theory \cite{ST}. More radically, one would like independence from any background topological spacetime structure, hence a sum over topologies \cite{topology}. This is realized in the simpler context of matrix models for 2d quantum gravity \cite{MM}. Whether one intends quantum gravity as a theory of quantum geometry or of quantum spacetime {\it tout court}, the usual spacetime structures (a smooth metric field plus possibly the spacetime manifold itself) should be reconstructed from more fundamental quantum degrees of freedom of a different nature. The candidates for such fundamental pre-geometric degrees of freedom differ from one approach to another, but several arguments (e.g. suggested by the thermodynamical properties of black holes \cite{sorkin}) support the idea that they have to be of a {\it discrete} nature. Discrete building blocks of a quantum spacetime are used in simplicial quantum gravity and matrix models. They are also found as a result of quantization even in a priori continuum approaches like LQG, in the form of spin networks and spin foams \cite{LQG, SF, zakopane}. 

\

\noindent Tensor models and group field theories \cite{GFT1,GFT2,tensorReview, GFT3, vincentTensor} (which we collectively label {\it tensorial group field theories} (TGFT) in this paper) are a fast growing approach with very promising features.
They incorporate many of the insights revealed by the above approaches. They rely on a fully background independent formalism in which quantum, discrete, pre-geometric building blocks are used to (hopefully) generate a quantum spacetime that is dynamical in both geometry and topology. Indeed, TGFTs are a generalization of matrix models. Rank-$d$ tensors (with $d>2)$ are the basic dynamical variables. Stranded diagrams dual to $d$-complexes are generated as Feynman diagrams of the theory. 
The tensor represents a ($d-1$)-simplex; the indices refer to its ($d-2$)-faces, and the pairing of tensor indices in the interaction represents the gluing of  several ($d-1$)-simplices to form a $d$-dimensional polyhedron (a $d$-simplex in the simplest models). This peculiar {\it combinatorial non-locality} of the interactions is a defining feature of the models.

\

\noindent In the simplest, purely combinatorial models of this type, referred to as {\it tensor models} and first introduced in the early 90's \cite{tensor}, the indices of the tensors take value in finite sets of dimension $N$. More structure to quantum states, action and dynamical amplitudes is the result of endowing the tensors with more interesting domain spaces. Proper TGFTs are obtained when these are chosen to be Lie group manifolds \cite{mikecarlo} or their dual Lie algebras \cite{aristidedaniele}, while maintaining the combinatorial structure of the interactions. The first examples of such richer models were introduced as a quantization of discrete BF theories \cite{boulatov,ooguri}, and were later refined to give a candidate quantization of 4d gravity in the context of spin foam models \cite{DPFKR}.  

\

\noindent When appropriate data are added to the tensorial field and to its action, TGFTs become in fact a way to define the dynamics of kinematical states of LQG. Boundary states of the theory assume the form of spin networks, and the Feynman amplitudes assume the form of spin foam models \cite{mikecarlo} (TGFTs become then a natural way to remove the dependence of the spin foam dynamics from a given cellular complex). In different variables, furthermore, the same amplitudes are expressed as simplicial gravity path integrals \cite{aristidedaniele}, as used in simplicial quantum gravity approaches. Beyond the relation with LQG and simplicial gravity, the group-theoretic data are crucial. On one hand they allow the use of mathematical tools otherwise unavailable (e.g. Peter-Weyl decomposition and recoupling theory, non-commutative Fourier transforms,  etc). On the other hand they endow the TGFT field, action and amplitudes with a much more transparent geometric interpretation. This is also a key for extracting effective continuum physics from the formalism. TGFTs allow then a new point of view on the dynamics of quantum spacetime as described in these approaches, resting on a {\it bona fide} quantum field theory framework. In this context, models for 4d quantum gravity have been developed, the most interesting ones being found in \cite{EPRL, BO-BC, BO-Immirzi}.

\

\noindent The field theory setting is crucial to addressing issues arising when a large number of pre-geometric degrees of freedom are involved, in particular to explore the continuum limit of TGFT models. Continuum spacetime and geometry have been suggested \cite{GFTfluid, lorenzoGFT, vincentTensor} to arise, in the TGFT context, through a phase transition (dubbed {\it geometrogenesis}, following \cite{graphity}), as happens in matrix models. The study of phase transitions in TGFTs, obviously, is best tackled using field theoretic tools, just as the analysis of symmetries \cite{joseph, GFTdiffeos, virasoro}, collective effects and effective dynamics, for example via mean field techniques \cite{danielelorenzo, danieleflorianetera,effHamilt}, or simplified models \cite{gfc}. Related work on the continuum limit of spin foam models and discrete gravity path integrals, from a lattice gauge theory perspective, is being carried out by Dittrich and collaborators \cite{bianca}.

\

\noindent Among the relevant quantum field theory tools, the renormalization group plays a pivotal role. The renormalization of TGFTs has been a subject of intense activity in recent years. It started with (single scale) power counting theorems for the Feynman amplitudes of tensor models and topological TGFTs \cite{GFTrenorm,lin,valentinmatteo,scaling}\footnote{Our understanding of divergences in TGFT models of 4d gravity remains still very limited \cite{divergences4dGFT}.}. 

\

\noindent
But ultimately renormalization of TGFTs requires a rather precise control of their combinatorial structure. For this, we can now rely on two crucial recent results. The first one is the introduction  of {\it colors} \cite{Gurau:2009tw} labelling a multiplet of TGFT fields. Colored graphs can be shown to encode the topology of general $d$-dimensional complexes \cite{crystallization}. Colors can be equivalently understood as labelling the ordered arguments of a \emph{single} un-symmetric tensorial field \cite{uncoloring}. Random un-symmetric tensorial fields have been found to have natural polynomial interactions based on  $U(N)^{\otimes d}$ invariance (where $N$ is again the size of the tensor) \cite{universality}. These interactions,
obtained by contractions of indices of same the position between the field and its complex conjugate are 
the ones considered in this paper. 

\

\noindent
The second key result is the TGFT analogue of the $1/N$ expansion of matrix models identified in \cite{large-N}. The TGFT perturbative expansion at large $N$
is dominated in any dimension $d$ by a particular class 
of triangulations of the $d$-sphere. They were further characterized and called {\it melons} in \cite{critical}. 
This $1/N$ expansion has allowed the first \emph{proofs} 
that a phase transition indeed occurs in simple tensor models \cite{critical, critical2}, and is a key tool 
in the analysis we perform in this paper.

\

\noindent
In the context of renormalization, one may distinguish two types of models: {\it ultralocal} ones, such as those apt for the description of topological BF theory, characterized by trivial kinetic operators (delta functions or simple projectors), and {\it dynamical} ones (first considered, with different motivations, in \cite{generalisedGFT}) with kinetic operators given by differential operators such as the Laplacian on the group manifold. In the first case, the models are non-trivial only thanks to the specific symmetries and other conditions imposed on the field and to the peculiar non-local nature of the interactions. In the second case, the propagator allows one to define scales and to launch a proper renormalization group flow. It is unclear whether
ultralocal models are rich enough to give a proper quantization of 4d gravity. Indeed there are indications \cite{josephvalentin} that even starting from ultralocal models one falls into dynamical models as soon as radiative corrections are considered, since the kinetic terms with Laplacian operators are required as counter-terms. Hence it is the second type of models that are considered in this paper.

\

\noindent
TGFTs are truly a new class of quantum field theories. They pose new challenges, in particular to renormalization, 
but offer also new promising features. The first examples of dynamical TGFTs {\it renormalizable} to all orders in perturbation theory have been identified \cite{tensor_4d, josephsamary,gelounlivine}. The most natural ones have been proved
asymptotically free \cite{josephsamary,gelounlivine,josephaf}. Asymptotic freedom, thanks to the wave function renormalization 
which is stronger in the tensorial context than in the scalar, vector or matrix case, may very well be a generic 
property of TGFTs. It makes them prime candidates for a geometrogenesis scenario which would be 
a kind of gravitational analog of quark confinement in QCD. 
Moreover TGFTs are accessible to rigorous \emph{constructive}
analysis in their dilute perturbative phase, through a constructive tool called the \emph{loop vertex expansion} \cite{Rivasseau:2007fr}.
This tool has been checked to apply to tensor models in \cite{Magnen:2009at} and to apply to positive even 
interactions of arbitrarily high order \cite{Rivasseau:2010ke}. TGFTs have therefore the potential for a non-perturbative and
rigorous analytic formulation that, to our knowledge, is yet lacking in some other approaches 
to quantum gravity. They also already include applications to domains of statistical 
physics such as dimers \cite{bonzomerbin} or spin glasses \cite{BGS} which are
quite far from the initial quantum gravity context.

\

\noindent
In spite of these recent successes, the renormalizable TGFTs analyzed so far do not yet have some additional symmetries which are needed for their Feynman amplitudes to be interpreted as discretized topological BF theories or (together with additional conditions) 4d quantum gravity. One needs to revise the renormalization tools introduced in \cite{tensor_4d} to include these symmetries, usually referred to also as closure conditions, and hereafter called the \emph{gauge invariance conditions}.
This is what we do in this work. More precisely, we tackle the issue of renormalizability for TGFT models with such geometric conditions with the rigorous tool of multi-scale analysis, to prepare the ground for future 
applications to gravity models.  Our main results are the following.

\

\noindent
We define the appropriate generalization of some of the notions of usual QFT which are key to the renormalization analysis, 
including:

\begin{itemize}

\item a new notion of {\it connectedness} for TGFT Feynman diagrams,

\item a new notion of (quasi-)locality, that we name {\it traciality}, adapted to the TGFT context; 

\item a new notion of {\it contraction of high subgraphs} (the ones that look local) \footnote{This contraction procedure can be also understood as a new coarse graining procedure for lattice gauge theory and discrete gravity, from the point of view of the Feynman amplitudes of the TGFT, seen as a discrete path integral (to be compared with the one used in \cite{bianca}).}.

\item a new notion of Wick ordering for general invariant interactions, which we name {\it melordering} (for {\it melonic Wick ordering}). It subtracts
their maximal melonic tadpoles, or \emph{melopoles}, which are all \emph{tracial}. Such melordering is a first step in the renormalization of any TGFT.
\end{itemize}

\noindent
We then consider the simplest examples of dynamical 4-dimensional TGFT with gauge invariance conditions 
for the Abelian $U(1)$ case. Their complex field depends on four $U(1)$ group elements (in configuration space). The propagator 
is the inverse of a Laplacian but with an added projection to represent $BF$-type gauge invariance conditions.
Interactions are given by arbitrary $U(N)^4$ invariant monomials.  

\

\noindent 
We perform the full multi-scale analysis of these models and we prove that they are 
{\it super-renormalizable} for any polynomial interaction of arbitrary order. Their only 
divergent diagrams are the melopoles. Hence melordering provides the renormalization. We
prove that the models with melordered interactions have a {\it finite renormalized series} at any order in perturbation theory. 
Therefore these $4$-dimensional models are the direct analogues in the tensor world of the $P(\phi)_2$ models of 
ordinary quantum field theory \cite{Simon}, which are also super-renormalizable for any polynomial interaction and in which Wick ordering provides
all the renormalization.

\

\noindent 
We conclude with an outlook into the non-Abelian case. It suggests that in $d=3$ a $SU(2)$-based model of the same type is {\it just renormalizable} with the 6-th order interactions considered in \cite{tensor_4d}.

%% file: sections/def_models.tex
 

\section{Definition of the models}

\subsection{Formal definition}

The class of theories we consider are tensorial group field theories of one single rank-$d$ complex tensorial field $\vphi(g_1 , \dots , g_d)$, whose arguments $g_\ell$ take value in a Lie group
$G$. In the spirit of \cite{universality}, the tensorial nature of the field $\vphi$ provides us with a natural notion of \textit{locality}, encoded by the fact that the interaction part of the action is a sum of \textit{tensor invariants}, as is the case in matrix models. Such invariants are obtained by convolution of a set of fields $\vphi$ and $\overline{\vphi}$, in such a way that the $k$-th index of a field $\vphi$ is
always contracted with the $k$-th index of a conjugate field $\overline{\vphi}$. Indeed, the result of such a convolution is a field polynomial invariant under $U(N)^{\otimes d}$, where N is the cut-off on $U(1)$ representations (playing the role of momenta). They are canonically represented by closed \textit{$d$-colored graphs}, constructed as follows: each field $\vphi$
(resp. $\vphib$) is represented by a white (resp. black dot), and each contraction of a $k$-th index between two fields is pictured as a line with \textit{color} label $k$ linking the two relevant
dots (see Figure \ref{tensinv}). Connected such graphs are called \textit{$d$-bubbles}. As usual in field theory, we will assume that the interaction is a sum of invariants which can be represented as such connected graphs,
which we call \textit{connected tensor invariants}. Therefore, we define the interaction part of the action as
\beq
S(\vphi , \vphib) = \sum_{b \in \cB} t_b I_b (\vphi , \vphib)\,, 
\eeq
where $\cB$ is a finite set of $d$-bubbles, $I_b$ the connected tensor invariant labelled by $b$, and $t_b \in \mathbb{C}$\footnote{Restrictions on the set of allowed values for $\{ t_b \}$
are necessary if we want $S$ to satisfy conditions such as reality or positivity, but at this stage we keep the discussion as general as possible.}. 

\begin{figure}[h]
\begin{center}
\includegraphics[scale=0.7]{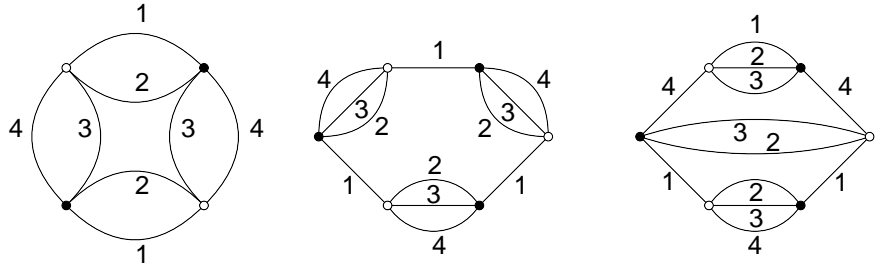}
\caption{Some connected tensor invariants in $d = 4$}
\label{tensinv}
\end{center}
\end{figure}
\

\noindent The kinetic part of the model is given by a Gaussian measure $\extd \mu_C (\vphi , \vphib)$, with covariance (propagator) $C$:
\beq
\int \extd \mu_C (\vphi , \vphib) \, \vphi(g_1 , \dots , g_d ) \vphib(g_1' , \dots , g_d' ) = C(g_1, \dots , g_d ; g_1' , \dots , g_d' )\,.
\eeq
The partition function is defined as 
\beq
\cZ = \int \extd \mu_C (\vphi , \vphib) \, \e^{- S(\vphi , \vphib )}.
\eeq
If $C$ itself is the kernel of a tensor invariant, as it is the case in simple i.i.d tensor models \cite{tensorReview, tensor}, or a projector, as it is the case in topological GFTs \cite{boulatov,ooguri,GFT1,GFT2,GFT3}, the model is called
\textit{ultralocal} and the usual field-theory notion of scale in terms of the spectrum of the covariance cannot be applied. In \cite{tensor_4d} the
first examples of renormalizable models were given, where the usual spacetime-based notion of scales is replaced by a more abstract notion, based on the
spectrum of $C$. Such a generalization is forced upon us by the background-independent nature of such models, at least if one wants to extend the scope of renormalization methods to such theories. Therefore, a covariance with a rich enough spectrum is necessary to the very definition of renormalizability, which depends in fact on the notion of scales. This was chosen in \cite{tensor_4d} to be $\widetilde{C} = \left( m^2 - \sum_{\ell = 1}^{d}
\Delta_\ell \right)^{-1}$, where $\Delta_\ell$ is the Laplace-Beltrami operator on $G$ acting on color-$\ell$ indices. Its kernel is an integral over a
Schwinger parameter $\alpha$ of products of heat kernels $K_\alpha$:
\beq
\widetilde{C}(g_1, \dots , g_d ; g_1' , \dots , g_d' ) =  \int_{0}^{+ \infty} \extd \alpha \, \e^{- \alpha m^2} \prod_{\ell = 1}^{d} K_{\alpha} (g_\ell g_\ell'^{\inv})\,.
\eeq
The parameter $\alpha$ is interpreted as a momentum scale, and can be sliced according to a geometric progression in order to perform a multi-scale analysis (see \cite{VincentBook}), as we do in the following.

\
If such a non-trivial propagator allows to deviate from ultralocality and hence address the question of renormalizability, it is still not satisfying from a discrete gravity (or lattice gauge theory) perspective. In TGFT models that aim at describing such theories, the TGFT fields have to satisfy some constraints which give them the interpretation of quantized $(d-1)$-simplices, and the Feynman amplitudes that of simplicial gravity path integrals or lattice gauge theory partition functions \cite{GFT1,GFT2,GFT3}. A common feature in any dimension $d$, is the so called
\textit{closure constraint}, which (in group representation) imposes invariance of the TGFT field under simultaneous (left) translations of its arguments:
\beq\label{gauge_inv1}
\forall h \in G \,, \qquad \vphi(h g_1, \dots , h g_d ) = \vphi(g_1, \dots g_d )\,.
\eeq
This is exactly the condition which, in the cellular complexes labelling TGFT amplitudes, allows to define a discrete (gravitational or Yang-Mills) connection, and gives them the general form of a lattice gauge theory amplitude\footnote{This is the reason why, despite the absence of gauge symmetry in the field theory sense, the invariance of the field (\ref{gauge_inv1}) is also referred to as gauge invariance condition: it is responsible for a lattice gauge symmetry at the level of each amplitude.}, when not that of a lattice gravity model, like in the Abelian case we will consider in the following.  In the present paper, it is our purpose to explore the consequences for renormalizability of the implementation of such a constraint. We therefore define
the new propagator $C$ as a group-averaged version of $\widetilde{C}$
\beq
C(g_1, \dots , g_d ; g_1' , \dots , g_d' ) = \int_{0}^{+ \infty} \extd \alpha \, \e^{- \alpha m^2} \int \extd h \prod_{\ell = 1}^{d} K_{\alpha} (g_\ell h g_\ell'^{\inv})\,,
\eeq
which is a way to ensure that only translation invariant degrees of freedom are propagated. This results in the Feynman amplitudes being written as integrals over discrete connections, i.e. as lattice gauge theory path integrals on the lattice given by the (dual of the) TGFT Feynman diagram. In Lie algebra representation, the same amplitude will then take the form of a $BF$-like simplicial path integral \cite{aristidedaniele}, while in representation space it will become a spin foam model \cite{mikecarlo}.

\subsection{Regularization and the question of renormalizability}

Models as defined in the previous sections are only formal, and plagued with divergencies. The two possible sources of divergencies are the vicinity of $0$ (UV) and the vicinity of $+
\infty$ (IR) for the Schwinger parameter $\alpha$ \footnote{We adopt a standard QFT terminology for the UV/IR distinction, adapted to the renormalization group flow; this should not be given any geometric interpretation at this stage. Notice that the corresponding nomenclature from the point of view of a {\it simplicial} gravity interpretation of the amplitudes is usually the opposite.}. In this paper we will discard IR divergencies, since the only explicit example we will work out will be defined on a compact group $G = \U(1)$, in
which case IR divergencies do not occur. Finally, as our main technical tool will be a multi-scale analysis, we choose a UV regulator compatible with a decomposition of the propagator into slices. The
latter goes as follows. We fix an arbitrary parameter $M > 1$, and decompose the integral over the Schwinger parameter $\alpha$ into slices $[ M^{ - 2 i} , M^{ - 2 (i - 1)}]$, where $i$ takes integer
values:
\bes
C_{0}(g_1, \dots , g_d ; g_1' , \dots , g_d' ) &=& \int_{1}^{+ \infty} \extd \alpha \, \e^{- \alpha m^2} \int \extd h \prod_{\ell = 1}^{d} K_{\alpha} (g_\ell h g_\ell'^{\inv})\,, \\
\forall i \geq 1\,, \qquad C_{i}(g_1, \dots , g_d ; g_1' , \dots , g_d' ) &=& \int_{M^{ - 2 i}}^{M^{ - 2(i-1)}} \extd \alpha \, \e^{- \alpha m^2} \int \extd h \prod_{\ell = 1}^{d} K_{\alpha} (g_\ell h
g_\ell'^{\inv})\,.
\ees  
The UV regulator $\rho$ is an upper bound in the sum over slices, defining the regularized propagator:
\beq  \label{decomposi}
C^{\rho} = \sum_{0 \le i \leq \rho} C_i \,.
\eeq 

\
The theory will be renormalizable if the UV regulator can be removed, and the resulting infinities absorbed in an appropriate finite set of local interactions. In the next section, we present the general form of the Feynman amplitudes obtained for this class of models, using the propagator and interactions (bubble invariants) which we have introduced above.

%% file: sections/amplitudes.tex
 

\section{Amplitudes}

\subsection{Feynman graphs and gauge symmetry}

The Feynman graphs of the theory are constructed from vertices or $d$-bubbles, 
supplemented by a set of lines (called of color $0$), representing propagators. 

The vertices are restricted to $U(N)^{\otimes D}$ connected tensor invariants, hence 
have a colored representation \cite{uncoloring}. Plugging this representation at every vertex, every Feynman 
graph of the theory has a unique underlying colored graph, which we also call its \textit{colored extension}.
Because of the complex nature of the $\vphi$ field, the color-$0$
lines in this colored extension (represented as dotted lines) must link black nodes to white nodes of the $d$-bubbles, each of these nodes being attached to at most $1$ line of color $0$. The underlying colored Feynman graphs are therefore all the $(d+1)$-colored graphs
with the restriction that only lines of color $0$ can be opened, i.e. external. An example is given in Figure \ref{coloredgraph}.

\begin{figure}[h]
\begin{center}
\includegraphics[scale=0.7]{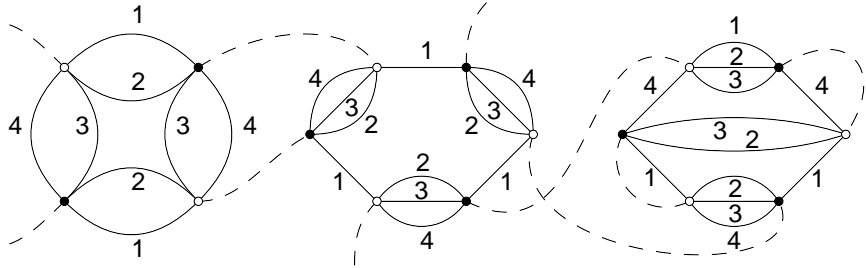}
\caption{A graph with three vertices, six (internal) lines, and four external legs}
\label{coloredgraph}
\end{center}
\end{figure}

In the following we will simply call \textit{graphs} the uncolored ones, their \textit{lines} being the lines of color $0$. 
We will also refer to the internal lines of the $d$-bubbles as the \textit{colored lines}.

\
At the formal level, the connected Schwinger functions of the theory are given by a sum over connected Feynman graphs:
\beq
\cS_N = \sum_{\cG \; \mathrm{connected}, N(\cG)= N} \frac{1}{s(\cG)} \left(\prod_{b \in \cB} (- t_b)^{n_b (\cG)}\right) \cA_\cG \,,
\eeq
where $N(\cG)$ is the number of external legs of a graph $\cG$, $n_b (\cG)$ the number of vertices of the type $b$, and $s(\cG)$ a symmetry factor. The Feynman rules for building $\cA_\cG$ are
straightforward: to each $d$-bubble 
of the type $b$ is associated an integral with respect to a measure given by the kernel of $I_b$, and to each color-$0$ line corresponds a propagator. Integrating the different heat kernels defining
the propagator, we can write it in a very 
similar way as a lattice gauge theory amplitude:
\begin{eqnarray}
\cA_\cG &=& \left[ \prod_{e \in L(\cG)} \int \extd \alpha_{e} \, e^{- m^2 \alpha_e} \int \extd h_e \right] 
\left( \prod_{f \in F (\cG)} K_{\alpha(f)}\left( \overrightarrow{\prod_{e \in \partial f}} {h_e}^{\epsilon_{ef}} \right) \right) \nn\\
&&\left( \prod_{f \in F_{ext}(\cG)} K_{\alpha(f)} \left( g_{s(f)}
\left[\overrightarrow{\prod_{e \in \partial f}} {h_e}^{\epsilon_{ef}}\right] g_{t(f)}^{\inv} \right) \right) \,.
\end{eqnarray}
In this formula, $\alpha(f) \equiv \sum_{e \in \partial f} \alpha_e$ is the sum of the Schwinger parameters appearing in the face $f$, and $\epsilon_{ef} = \pm 1$ is determined by the orientation of $e$
with respect to an arbitrary orientation of the faces. The faces are split into closed ($F$) and opened ones ($F_{ext}$), $g_{s(f)}$ and $g_{t(f)}$ denoting boundary variables in the latter
case, with functions $s$ and $t$ mapping open faces to their ``source" and ``target" boundary variables. 

\
This amplitude is invariant under a group action acting on the vertices (i.e. the $d$-bubbles). For any assignment of group elements $(g_v) \in G^{V(\cG)}$, the integrand of the amplitude is invariant
under:
\beq
h_e \mapsto g_{t(e)} h_e g_{s(e)}^{\inv}\,,
\eeq
where $t(e)$ (resp. $s(e)$) is the target (resp. source) vertex of an (oriented) edge $e$, and one of the two group elements is trivial for open lines. Because this is a symmetry of the integrand
itself, it can be gauge-fixed following the standard prescription of \cite{FLouapre}. When $\cG$ is connected, this amounts to set $h_e = \one$ for all lines $e$ in a maximal tree $\cT$. This
gauge symmetry is a very important feature of the models considered in this paper, that will require significant modifications of the multi-scale analysis of \cite{tensor_4d}.

\subsection{Multi-scale decomposition of the amplitudes}

Using the multi-scale decomposition \eqref{decomposi}, any graph is written as a sum over scale attributions
$\mu = \{ i_e\}$, where $i_e$ runs over all integers (smaller than $\rho$) for every line $e$.

\begin{eqnarray}
\cA_{\cG }  &=&  \sum_{\mu }\cA_{\cG , \mu} , \nn\\
\cA_{\cG , \mu}  &=& \left[ \prod_{e \in L(\cG)}  \int_{M^{ - 2 i}}^{M^{ - 2(i-1)}}  \extd \alpha_{e} \, e^{- m^2 \alpha_e} \int \extd h_e \right] 
\left( \prod_{f \in F (\cG)} K_{\alpha(f)}\left( \overrightarrow{\prod_{e \in \partial f}} {h_e}^{\epsilon_{ef}} \right) \right) \nn\\
&&\left( \prod_{f \in F_{ext} (\cG)} K_{\alpha(f)} \left( g_{s(f)}
\left[\overrightarrow{\prod_{e \in \partial f}} {h_e}^{\epsilon_{ef}}\right] g_{t(f)}^{\inv} \right) \right) \,.
\end{eqnarray}

The strategy of the multi-scale expansion is to replace the complicated expression for the propagators by a simpler bound
which captures their power-counting, and to integrate the variables $h_e$ without loosing any such power-counting, hence
any powers of the $M^{i}$ type. But to implement it and to perform the renormalization 
of this new type of models we have to revise some graph-theoretical notions and adapt them to our new context.

\subsection{Dipole moves and reduced graphs}

A central notion in colored tensor models and GFTs is that of dipole contractions. They were key to the discovery of the $1/N$-expansion \cite{large-N}, as well as to answering more specific
topological questions \cite{scaling, jimmy,francesco}, essentially because they are the counterparts of Pachner moves in the colored context \cite{crystallization}. The main appeal of these moves
is that they allow to reduce the combinatorial complexity of colored graphs, while retaining topological properties of their dual simplicial complexes. In the present paper however we will have a
somewhat different approach as we will use dipole contractions as a way to consistently delete faces, and implement a contraction scheme for quasi-local graphs. Topological considerations will
therefore be essentially irrelevant in what follows. For this reason, we will not distinguish \textit{degenerate} from \textit{non-degenerate}
dipoles, and will use the generic word \textit{dipole} for both. Finally, since lines of color $0$
are the only dynamical entities in our framework, we will also only consider dipoles of a special kind: those which have an internal color-$0$ line.
Their precise definition is the following:
\begin{definition}
Let $\cG$ be a graph, and $\cG_c$ its colored extension. For any integer $k$ such that $1 \leq k \leq d+1$, a $k$-dipole is a line of $\cG$ whose image in $\cG_c$ links two nodes $n$ and
$\overline{n}$ which are connected by exactly $k - 1$ additional colored lines (see Figure \ref{dipo}).
\end{definition}

\begin{figure}
\begin{center}
\includegraphics[scale=0.7]{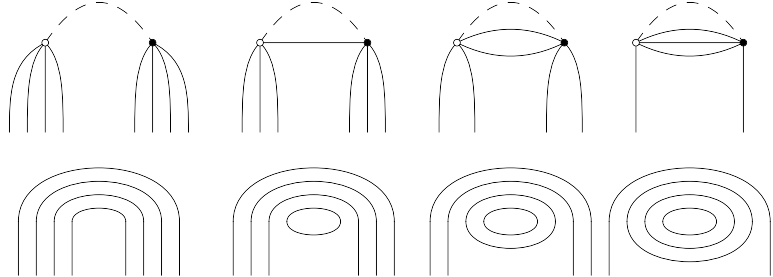}
\caption{$k$-dipoles from $k = 1$ (top, left) to $k = 4$ (top, right) in $d = 4$, and their faces (below)}
\label{dipo}
\end{center}
\end{figure}

With the terminology of previous works \cite{tensorReview}, a $k$-dipole of $\cG$ is nothing but a line whose image is internal to a $k$-dipole of $\cG_c$.
We now define the contraction operation.
\begin{definition}
Let $\cG$ be a graph, and $\cG_c$ its colored extension. The contraction of a $k$-dipole $d_k$ is an operation in $\cG_c$ that consists in:
\begin{enumerate}[(i)]
 \item deleting the two nodes $n$ and $\overline{n}$ linked by $d_k$, together with the $k$ lines that connect them;
 \item reconnecting the resulting $d - k$ open legs according to their colors.
\end{enumerate}
We call $\cG_c / d_k$ the resulting colored graph, and $\cG / d_k$ its pre-image.
\end{definition}

$1$-dipole contractions play a prominent role in colored tensor models an GFTs, because they implement the topological notion of connected sum of $d$-bubbles. This remains true in our context, the
relevant move being the contraction of a $1$-dipole which is not a tadpole, that is to say what is usually called a non-degenerate $1$-dipole. Interestingly, the contraction of a full set of such
$1$-dipoles is intimately related to the gauge-fixing procedure sketched before. In a connected graph $\cG$, a full contraction is obtained by successively contracting a maximal tree of lines $\cT$.
But we also know that in the Feynman amplitude of $\cG$, the group elements associated to this tree can be set to $\one$. Incidentally, the purely combinatorial notion of contraction of lines of $\cT$
is nothing but the result of trivial convolutions in the amplitude $\cA_\cG$. The only difference between $\cA_\cG$ and $\cA_{\cG/\cT}$, where $\cG / \cT$ denotes the fully contracted graph, is a set
of simple integrals with respect to Schwinger parameters, while their integrands have exactly the same structure. This observation will be crucial to the multi-scale analysis, and most of the
discussion will therefore focus on \textit{reduced graphs} $\cG / \cT$, with respect to a suitably chosen spanning tree.

\
All $k$-dipoles with $k > 1$ are necessarily tadpole lines, and contain $(k - 1)$ internal faces (see Figure \ref{dipo}). Among them, the $d$-dipoles will be the focus of special interest, and will be called \textit{melonic
tadpoles}, or simply \textit{melons}. In ultralocal tensor models and GFTs, these structures have been shown to govern the leading order, and henceforth the critical behaviour.

%% file: sections/connectedness.tex
 

\section{Connectedness and quasi-locality}

\subsection{Contraction of connected subgraphs}

We now give general definitions which are crucial to the implementation of a multi-scale analysis, and hence to the understanding of the precise structure of divergencies. 

Probably
the most important notion is that of \textit{quasi-local graphs}, that is \textit{connected subgraphs} which, from the point of view of their external legs, \textit{look} local. 
As a tensorial field
theory differs substantially from usual field theory, we need to reconsider in detail the notions of subgraph, connectedness, and contraction of subgraphs. 

A major role
will be played by the faces of the graph, which is where the curvature of the discrete connection introduced by the new gauge invariance condition is assigned. These faces are followed easily by drawing the 
\textit{colored extension} $\cG_c$ of the graph $\cG$ \cite{universality}. Faces (of fixed color $i$) of  $\cG$ are the alternating circuits of lines of color 0 and  $i$ in  $\cG_c$, and can be either closed 
(internal) or open (external). Rather than the usual incidence matrix $\epsilon_{ev}$ between lines and vertices of ordinary graph theory, it is the incidence matrix of lines and closed faces $\epsilon_{ef}$ in $\cG$ 
which plays the leading role in group field theory \cite{GFTrenorm, lin, tensor_4d}. To define this matrix one needs an orientation of both the lines and the faces. 
It is $+1$ if the face $f$ goes through line $e$ with the same orientation, $-1$ if the face $f$ goes through line $e$ with 
opposite orientation and 0 otherwise. The colored structure ensures absence of ``tadfaces", i.e. faces which pass several times through the same line.

\
We start with the notion of subgraph. In ordinary graph theory a subgraph of a graph $G$ is most conveniently defined as a {\emph{subset S of lines}} of $G$, so that a graph with $L$ lines has exactly $2^L$
subgraphs. Such a subset of lines is then completed canonically by adding the vertices attached to the lines and the external lines, also called ``legs". The latter are defined by first cutting in the middle all
lines of $G-S$. Legs of $S$ then correspond 
either to true legs of $G$ attached to vertices of $S$ or to half-lines of $G-S$ attached to the vertices of $S$. Finally ordinary connectedness of $H$ can be defined
through examining the ordinary incidence matrix $\epsilon_{ev}$ of $H$; the connected components of $H$  are the maximal factorized rectangular blocks of this matrix. Hence 
elementary connections between lines come from their common attached vertices.

Recalling that a tensorial graph $\cG$  has (0-colored) internal lines, external legs, ($d$-bubbles) vertices, and faces, 
the definition of a subgraph for TGFTs is a natural generalization of the ordinary definition.

\begin{definition}
A \textit{subgraph} $\cH$ of a graph $\cG$ is a subset of  lines of $\cG$, hence $\cG$ has exactly $2^{L(\cG)}$ subgraphs. 
$\cH$ is then completed by first adding the vertices that touch its lines. 
The faces closed in $\cG$ which  pass only through lines of $\cH$ form the set of \textit{internal} faces of $\cH$. The 
external faces of $\cH$ are the maximal open connected pieces of either open or closed faces of $\cG$ that pass through lines of $\cH$. 
Finally only the external legs touching the vertices of $\cH$ that contain external faces of $\cH$ are considered external legs of $\cH$.
\end{definition}
To understand the open faces of  $\cH$, the best is first to "cut", as in the ordinary case, all lines of $\cG \setminus \cH$ in the middle. This breaks the closed or open faces of $\cG$ into pieces.
Those remaining open pieces that belong entirely to $\cH$ are its open faces. Hence a closed or open face of $\cG$ which passes through lines of $\cH$ can 
generate several open faces of $\cH$.

We denote $L(\cH)$ and $F(\cH)$ the set of lines and internal faces of $\cH$, and $N(\cH)$ and $F_{ext}(\cH)$ the set of external legs and external faces of $\cH$. When no confusion is possible we also write $L$, $F$ etc for the cardinality 
of the corresponding sets.

\begin{definition}
The connected components of a subgraph $\cH$ are defined as the subsets of lines of the maximal factorized rectangular blocks of its $\epsilon_{ef}$ incidence matrix.
\end{definition}
Hence elementary connections between lines now come from their common internal faces. 
Notice also that a connected graph $\cH$ has always a connected colored extension $\cH_{c}$, but the converse is not true.

\medskip

\noindent{\bf 
Example:} The tadpole graph $\cG$ of the Appendix (Figure \ref{melop6}) has three lines, one vertex, no external legs and ten internal (closed) faces. It has eight subgraphs. 
The subgraph $S_1= \{l_1\}$ has one line, one vertex, three internal (closed) faces, one external face and two external legs (the two halves of line $l_2$). The same is true
for the subgraph $S_3= \{l_3\}$. The subgraph $S_2= \{l_2\}$ has one line, one vertex, two internal (closed) faces, two external faces and four external legs.
The subgraph $H= \{l_1, l_3\}$ has two lines, one vertex, two external legs, 6 closed faces and two external faces, but it is not connected (although its colored extension is).
It has two connected components $S_1$ and $S_3$ which although having a vertex and two external lines in common are not connected through their faces; the
2 by 6 incidence matrix of this subgraph factorizes in two 1 by 3 blocks, those of $S_1$ and $S_3$. 

\

The third notion we need to extend to tensorial group field theory is that of {\it contraction of a subgraph}. In graph theory, contracting a line simply means shrinking it until its two end vertices are identified.
In our situation this does not make sense anymore, since the vertices have an internal color structure that prevents us from simply concatenating them. 
Instead, one can identify the two end nodes in the colored representation of the graph, which in turn should be interpreted as the identification, color by color, of the group variables they are attached too. 
This naturally leads to the realization of line contractions as dipole contractions. In the
following, we will therefore simply call \textit{contraction of the line $e$} the contraction of the canonically associated dipole $e$.
This definition extends to contractions of subgraphs.
\begin{definition}
We call contraction of a subgraph $\cH \subset \cG$ the successive contractions of all the lines of $\cH$. The resulting graph is independent of the order in which the lines of $\cH$ are contracted,
and is noted $\cG / \cH$.
\end{definition}
\begin{proof}
To confirm that this definition is consistent, we need to prove that dipole contractions are commuting operations. Consider two distinct lines $e_1$ and $e_2$ in a graph $\cG$, and call $\cH$ the
subgraph made of 
$e_1$ and $e_2$. We distinguish three cases.
\begin{enumerate}[(i)]
 \item $\cH$ is disconnected. This means that $e_1$ and $e_2$ are part of two independent dipoles, with no colored line in common, and the two contraction operations obviously commute.
 \item $\cH$ is connected, and none of its internal faces contain both $e_1$ and $e_2$. This means that $e_1$ and $e_2$ are contained in two dipoles $d_1$ and $d_2$, such that for each
color $i$, at most one line of color $i$ connects $d_1$ to $d_2$. Hence contracting $d_1$ (resp. $d_2$) does not change the nature of the dipole in which $e_2$ (resp. $e_1$) is contained. 
So here again, $d_1$ and $d_2$ are local objects which can be contracted independently. 
 \item $\cH$ has $q \geq 1$ internal faces containing both $e_1$ and $e_2$. In this case, the contraction of $e_1$ (resp. $e_2$) changes the nature of the dipole in which $e_2$ (resp. $e_1$) is
contained: $q$ internal faces are added to it. However, when contracting the second tadpole, these faces are deleted, so for any order in which the contractions are performed, 
all the internal faces are deleted. As for the external faces, the situation is the same as in the previous case, and the two contractions commute.
\end{enumerate}
\end{proof}
We can finally give a more global characterization of the contraction operation. 
\begin{proposition}
Let $\cH$ be a subgraph of $\cG$, and $\cH_c$ its colored extension. The contracted graph $\cG / \cH$ is obtained by:
\begin{enumerate}[(a)]
 \item Deleting all the internal faces of $\cH$;
 \item Replacing all the external faces of $\cH_c$ by single lines of the appropriate color.
\end{enumerate}
\end{proposition}
\begin{proof}
We prove this by induction on the number of lines in $\cH$. If $\cH$ contains one single line, then it is a dipole, and the proposition is true according to the very definition of a dipole
contraction. Now, suppose that $\cH$
is made of $n>1$ lines, $n-1$ of them being contained in the subgraph $\cH_0 \subset \cH$, and call the last one $e$. The set of internal faces in $\cH$ decomposes into several subsets. The faces
which are internal to $\cH_0$ are deleted 
by hypothesis when contracting $\cH_0$. Those common to $\cH_0$ and $e$ become internal dipole faces once $\cH_0$ is contracted, so they are deleted when $e$ is contracted, and the same is of course
true for the remaining internal faces which have $e$ as single line. The same distinction of cases applied to external faces of $\cH_c$ allows to prove that they are replaced by single lines of the
appropriate color, which achieves
the proof. 
\end{proof}

\

Now that we have these notions at our disposal, we can address the question of how contracting connected subgraphs within bigger graphs affect their properties. In usual graph theory, the
number of connected components in a graph 
$\cG$ is not affected by the contraction of a subgraph $\cH$. This property is crucial to scale analysis in field theory, which relies on contractions of specific connected divergent
subgraphs within bigger connected graphs. The new notion of connectedness in TGFT should be conserved under contraction of (at least) some class of connected subgraphs. The key
difference with usual field theory is that now the 
contraction of a line does not generically conserve connectedness. Indeed it is easy to check that

\begin{proposition}\label{disconnected}
\begin{enumerate}[(i)]
 \item For any connected graph $\cG$, if $e$ is a line of $\cG$ contained in a $d$-dipole, then $\cG / e$ is connected.
 \item For any $1 \leq q \leq d - k + 1$, there exists a connected graph $\cG$ and a $k$-dipole $e$ such that $\cG / e$ has exactly $q$ connected components. 
\end{enumerate}
\end{proposition}

\

\noindent{\bf Localization/Contraction Operators}

\noindent Combinatorial contractions of graphs are associated to localization operators acting on amplitudes. Let us consider a graph $\cG$, and a connected subgraph $\cH \subset \cG$. We define an operator
$\tau_\cH$ 
by its action on the integrand of $\cG$. The amplitude $\cA_\cG$ is of the form:
\bes\label{ampl}
\cA_\cG &=& \left[ \prod_{e \in L(\cH)} \int \extd \alpha_{e} \, e^{- m^2 \alpha_e} \int \extd h_e \right] 
\left( \prod_{f \in F(\cH)} K_{\alpha(f)}\left( \overrightarrow{\prod_{e \in \partial f}} {h_e}^{\epsilon_{ef}} \right) \right)
 \nn \\ 
 && \left[ \prod_{e \in N(\cH)} \int \extd g_e \right] 
 \left( \prod_{f \in F_{ext}(\cH)} K_{\alpha(f)} \left( g_{s(f)}
\left[\overrightarrow{\prod_{e \in \partial f}} {h_e}^{\epsilon_{ef}}\right] g_{t(f)}^{\inv} \right) \right)
\times \cR_{\cG \setminus \cH}\left(\{g_{s(f)}, g_{t(f)}\}\right)\,,
\ees
where $\cR_{\cG \setminus \cH}$ only depends on $g$ variables appearing in the external faces of $\cH$. We then define $\tau_\cH$ as:
\beq
\tau_\cH \cR_{\cG \setminus \cH}\left(\{g_{s(f)}, g_{t(f)}\}\right) \equiv \cR_{\cG \setminus \cH}\left(\{g_{s(f)}, g_{s(f)}\}\right)\,,
\eeq 
that is by moving all target variables of the external faces of $\cH$ to the sources. This definition is motivated by the fact that when the parallel transports inside 
$\cH$ are negligible, holonomies along external faces can be well approximated by directly connecting the two points at the boundary of $\cH$.

\begin{proposition} 
 Let $\cH \subset \cG$ be a connected subgraph. The action of $\tau_{\cH}$ on $\cA_\cG$ factorizes as:
\beq
\tau_\cH \cA_\cG = \nu_{\rho}(\cH) \cA_{\cG / \cH} \,,
\eeq
where $\nu_{\rho}(\cH)$ is a numerical coefficient depending on the cut-off $\rho$, and given by the following integral:
\beq
\nu_{\rho}(\cH) \equiv  \left[ \prod_{e \in L(\cH)} \int_{M^{-2 \rho}}^{+ \infty} \extd \alpha_{e} \, e^{- m^2 \alpha_e} \int \extd h_e \right] 
\left( \prod_{f \in F(\cH)} K_{\alpha(f)}\left( \overrightarrow{\prod_{e \in \partial f}} {h_e}^{\epsilon_{ef}} \right) \right)\,.
\eeq
\end{proposition}
\begin{proof}
 Applying $\tau_\cH$ in equation (\ref{ampl}), one remarks that heat kernels associated to external amplitudes can readily be integrated with respect to the variables $g_{t(f)}$. These integrals give
trivial contributions, 
thanks to the normalization of the heat kernel. We are therefore left with an integral over the internal faces of $\cH$, giving $\nu_{\rho}(\cH)$, times an amplitude which is immediately identified to be that
of the contracted graph $\cG / \cH$.
\end{proof}
{\bf Remark. } One can also use the same kind of factorization for the action of $\tau_\cH$ on $\cA_{\cG , \mu}$, in which case we will use the notation $\nu_\mu (\cH)$: $
\tau_\cH \cA_{\cG , \mu} = \nu_{\mu}(\cH) \cA_{\cG / \cH , \mu}$.

An illustration of the factorizability property is given in Fig. \ref{melo1}. This definition of the localization operators will be sufficient for the renormalization of logarithmic divergences, hence for all cases we will consider. 
As usual, the renormalization of power-like divergencies in more complicated models will require to push to a higher order the Taylor expansion around localized terms.

\subsection{High subgraphs, contractiblity, and traciality}

The graph-theoretic tools that have just been introduced allow to describe the general structures which will be relevant to the multi-scale analysis of tensorial group field theories. As compared to ordinary field theories, renormalizability 
of tensorial group field theories involves non-trivial refinements, which have to do with the conditions at which a high subgraph can be seen as a quasi-local effective object. In the usual $\varphi^{4}$ model
for example, quasi-locality
depends only on the separation of internal scales from external ones. We know already that this property is lost in tensorial models like the one in \cite{tensor_4d}, where high subgraphs can give rise
to disconnected effective 
invariants. We are going to see that the non-trivial projection we added in the propagator complicates the situation even more. The challenge is to keep the notion of scales firmly locked into the propagators:
defining scales for faces is a non-starter, since, contrary to propagators, the combinatorial structure of faces vary from graph to graph (it is a global construction in terms of Feynman rules), hence have no QFT meaning.

\

Let $\cG$ be an open graph, $\mu = \{ i_e \}$ a scale attribution, and $\cA_{\cG , \mu}$ the corresponding amplitude. As usual, we define high subgraphs according to internal and external scales. The
notion of external leg itself has been chosen to be compatible with the refined notion of connectedness we adopted, hence the correct notion of external scale is that of external legs.
\begin{definition}
\begin{enumerate}[(i)]
\item Given a subgraph $\cH \in \cG$, one defines:
\beq
i_{\cH}(\mu) = \inf_{e \in L(\cH)} i_e (\mu)\,, \qquad e_{\cH}(\mu) = \sup_{e \in N(\cH)} i_e (\mu)\,.
\eeq
\item A subgraph $\cH \in \cG$ is called \textit{high} if it is {\emph{connected}} and $e_{\cH}(\mu) < i_{\cH}(\mu)$. 
\end{enumerate}
\end{definition}
As usual the key to successful \emph{renormalization} is an approximate locality property of all the high \emph{divergent}  subgraphs seen from their external legs.

Our definition that a line is external to a subgraph $\cH$ if and only if there exists one colored line linking it to a line of $\cH$
reflects the need to define this approximate locality with respect to properties of the faces. This is also a consequence of the fact that, in our discrete gauge theoretic framework, 
the ``size" of a given group element can only be compared to the size of the other elements appearing in the same faces. 
 
\

In scalar field theory, any high subgraph (no matter whether divergent or not) systematically \textit{looks local} and can be interpreted as an effective interaction. In matrix models
such as \cite{GW}, this property is already lost for general high graphs and holds only for regular subgraphs (planar with a single external boundary), a subclass which fortunately contains all the divergent ones.
In the tensorial group field theories considered in this paper, this property is further restricted, but still applies to all divergent graphs. This is a very important result, since it is the one that makes renormalization possible, ultimately.

\

The first source of complications, which lies in the non-conservation of connectedness under dipole contractions 
(proposition \ref{disconnected}), was already identified in \cite{tensor_4d}.

\

The second complication is a new feature of the models considered here, and is due to the additional projection introduced in the propagator. Indeed, high propagators in these models do not
approximate the identity operator, but rather the projector on translation invariant fields. At the level of the amplitudes, this tells us that in high subgraphs, holonomies around closed faces are
very close to $\one$, which in itself does not say anything about the value of individual group elements $h_e$. But a dipole contraction is a good approximation only when the variables $h_e$ themselves
are close to $\one$, hence the tension.

\

A first remark is that gauge-symmetry allows to set $h_e = \one$ for $1$-dipoles between different vertices and
to contract them, reducing all loop lines to tadpole lines of a - big but single - connected invariant vertex, usually called a ``rosette" in the matrix context. In this rosette
the $k$-dipoles with $k>1$ are not a problem, as their internal faces automatically peak $h_e$ around $\one$, 
allowing to perform their contraction; but remark that this operation may disconnect the rosette into several connected components. 
Finally in each of these components or sub-rosettes the issue reduces to the study of $1$-dipoles. Typically a set of high such 1-dipoles 
peaks some $h_e$ around $\one$, but not all in general. 

\
To disentangle the loss of tensorial invariance from the loss of connectedness, we define two classes of subgraphs, the {\it contractible} and the {\it tracial} subgraphs. 
\begin{definition}
Let $\cG$ be a connected graph, and $\cH$ one of its connected subgraphs. 
\begin{enumerate}[(i)]
\item If $\cH$ is a tadpole, $\cH$ is \textit{contractible} if, for any group elements assignment $(h_e)_{e \in L(\cH)}$:
\beq
\left( \forall f \in F(\cH)\,, \;  \overrightarrow{\prod_{e \in \partial f}} {h_e}^{\epsilon_{ef}} = \one \right) \Rightarrow \left( \forall e \in L(\cH) \, , \; h_e = \one \right)\,.
\eeq
\item In general, $\cH$ is contractible if it admits a spanning tree $\cT$ such that 
$\cH / \cT$ is a contractible tadpole.
\item $\cH$ is \textit{tracial}\footnote{We thank Adrian Tanasa for suggesting this name.} if it is contractible and the contracted graph $\cG / \cH$ is connected.
\end{enumerate}
\end{definition}

A contractible graph is therefore a subgraph on which any flat connection is trivial up to a gauge transformation. Note that this gauge freedom is what makes the contraction with respect to a spanning tree an essential feature of the definition. On the other hand, the notion of traciality is independent of the choice of tree, as it is a statement about $\cG / \cH$, in which all internal lines of $\cH$ have been contracted.

\
In the multi-scale effective expansion, high divergent subgraphs give rise to effective couplings. To apply this procedure in our context, such subgraphs need to be
tracial, or at least contractible. Traciality ensures that the divergence of a high subgraph can be factorized into a divergent coefficient times a \emph{connected}
invariant. For  high divergent subgraphs which are contractible but not tracial, a factorization of the divergences is still possible, but 
in terms of disconnected invariants; these have been called anomalous terms in \cite{tensor_4d}. It is not clear yet whether this is a major issue and how these
anomalies should be interpreted physically, but in the models considered below all the divergent high subgraphs are tracial. Indeed
we already noticed that any $k$-dipole with $k >1$ is contractible, and that any $d$-dipole is tracial, as its contraction also preserves connectedness. 
These two facts combined provide us already with an interesting class of tracial subgraphs. We call them {\it melopoles} because they combine the idea of melonic graphs
and tadpoles. The high divergent graphs considered in the models of this paper will all be melopoles.

\begin{definition}
In a graph $\cG$, a \emph{melopole}
is a connected single-vertex subgraph $\cH$ (hence $\cH$ is made of tadpole lines attached to a single vertex in the ordinary sense), 
such that there is at least one ordering (or ``Hepp's sector") of its $k$ lines as $l_1, \cdots , l_k$ such that $ \{l_1 ,  \cdots , l_{i} \} / \{l_1 , \cdots , l_{i-1} \} $ is a $d$-dipole for $1 \le i \le k$.
\end{definition}

\begin{figure}
\begin{center}
\includegraphics[scale=0.7]{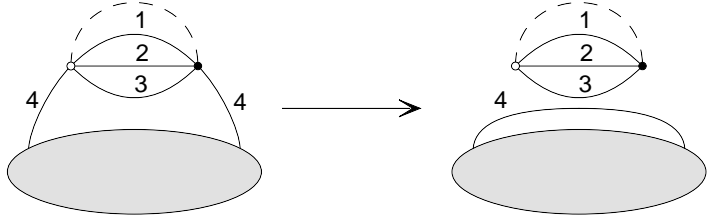}
\caption{A single-line melopole (left), and the result of its contraction}
\label{melo1}
\end{center}
\end{figure}

The simplest melopole has just one line and is shown in Figure \ref{melo1}. Its contraction within a connected graph (grey blob) results in a connected graph times a coefficient (of which a graphical representation is given).
In the example discussed in the Appendix (Figure \ref{melop6}), the subgraphs $\cH = \{l_1\}$ or $\cH = \{l_1, l_2\}$ are melopoles; the subgraphs $\cH = \{l_2\}$ and $\cH = \{l_1, l_3\}$ are not (the last one because it is not connected).
\begin{proposition}
Any melopole is tracial.
\end{proposition} 
\begin{proof}
Obviously it is contractible; and connectedness cannot be lost at any stage if one contracts in the order of the correct Hepp's sector.
\end{proof}

\
Now that we have clarified these notions and obtained a first class of very simple tracial graphs, namely the melopoles, we
can proceed with the multislice analysis and renormalization of the simplest tensor group field theories. 
The Gallavotti-Nicol\'o tree is an abstract tree encoding the inclusion order on high subgraphs of $(\cG, \mu)$, and hence an expansion of $\cA_{\cG , \mu}$ in terms of
effective vertices. For any $i \in \llbracket 0 , \rho \rrbracket$, we define $\cG_{i}$ as the subgraph made of all lines of $\cG$ with scale higher or equal to $i$. We further call $k(i)$ its number
of connected components, and $\{ \cG_{i}^{(k)} \,, k \in \llbracket 1 , k(i) \rrbracket \}$ its connected components. This latter set is exactly the set of high subgraphs. Two high
subgraphs are either included into one another, or disjoint. They therefore form what is called an \textit{inclusion forest}, 
because their inclusion relations can be represented by an abstract forest. When $\cG$ itself is
connected, and because it is also high by convention, the inclusion forest is 
actually an \textit{inclusion tree}: the Gallavotti-Nicol\'o tree.

%% file: sections/multiscale.tex
 

\section{Multi-scale analysis of Abelian models}
 
 \subsection{Propagator bounds}
 
 The explicit form of the heat-kernel at time $\alpha$ on $\U(1)$ is
 \beq
 K_{\alpha}(\theta) = \frac{\e^{- \frac{1}{4 \alpha} \theta^2 }}{\sqrt{\alpha}} \left( 1 + 2 \sum_{n = 1}^{\infty} \e^{- \frac{\pi^2 n^2}{\alpha}} \cosh\left( \frac{n \pi}{\alpha} \theta \right) \right)\,,
 \eeq
 so that the propagator can explicitly be written as:
 \beq
 C(\theta_1, \dots , \theta_d ; \theta_1' , \dots , \theta_d' ) = \int_{0}^{+ \infty} \extd \alpha \, \frac{\e^{- \alpha m^2}}{\alpha^{d/2}}  \int_{0}^{2 \pi} \extd \lambda \, \e^{- \frac{1}{4 \alpha} \sum_\ell (\theta_\ell - \theta_\ell' + \lambda)^{2} }
 T(\alpha ; \theta_1 - \theta_1' + \lambda , \dots , \theta_d - \theta_d' + \lambda)\,,
 \eeq
 with 
 \beq
 T(\alpha ; \theta_1 , \dots , \theta_d ) \equiv \prod_{\ell = 1}^{d} \left( 
 1 + 2 \sum_{n = 1}^{\infty} \e^{- \frac{\pi^2 n^2}{\alpha}} \cosh\left( \frac{n \pi}{\alpha} \theta_\ell \right) \right)\,.
 \eeq
 
 One can then prove generic bounds for the sliced propagators $C_i$, on which the whole power-counting will rely. Such bounds have been computed in the model of \cite{tensor_4d}, and immediately imply the following:
 
 \begin{proposition}
 There exist constants $K > 0$ and $\delta > 0$, such that for all $i \in \mathbb{N}$:
 \bes\label{propa_bound} 
 C_i (\theta_1 , \dots , \theta_d ; \theta_1' , \dots , \theta_d') &\leq& K M^{(d-2) i } \int \extd \lambda \,
 \e^{- \delta M^{i} \sum_\ell |\theta_\ell - \theta_\ell' + \lambda| }\,, \\
\forall \ell \in \llbracket 1 , d \rrbracket \,, \qquad\frac{\partial}{\partial \theta_\ell} C_i (\theta_1 , \dots , \theta_d ; \theta_1' , \dots , \theta_d') &\leq& K M^{(d-1) i } \int \extd \lambda \,
 \e^{- \delta M^{i} \sum_\ell |\theta_\ell - \theta_\ell' + \lambda| }\,, \label{deriv_bound}
 \ees
 \end{proposition}

The bound on the derivative of $C_i$ is generalizable to any number of derivatives, but we will only use the one we have just stated. For Abelian compact Lie groups of dimension $D$, $\theta$'s and $\lambda$ are $D$-dimensional and, for example, the first bound becomes
 \beq\label{propa_boundD} 
 C_i (\vec \theta_1 , \dots , \vec\theta_d ; \vec\theta_1' , \dots , \vec\theta_d') \leq K M^{(dD-2) i } \int \extd \vec\lambda \,
 \e^{- \delta M^{i} \sum_e |\vec\theta_e - \vec\theta_e' + \vec\lambda| }\,.
 \eeq

\subsection{Power-counting}
 
 \subsubsection{Power-counting in a slice}
  
The divergence degree of Abelian TGFT subgraphs in a single slice and
 with a heat kernel regularization has been 
 established and analyzed in \cite{lin}. For an Abelian compact group of dimension $D$ it gives
 \beq  \omega (\cH)  = -2 L (\cH) + D(F(\cH) - r(\cH)) 
 \eeq
where $r$ is the rank of the $\epsilon_{ef}$ incidence matrix of $\cH$. The $\theta$ integrations transform the $(dD-2)L$ 
into the $-2 L + D F$ term, whence the $\lambda$ integrals (absent in \cite{tensor_4d}) add the $r$ term. 

The factor $-2L$ is independent of both $D$ and $d$ for a $( m^2 - \sum_{\ell = 1}^{d}
\Delta_\ell )^{-1}$
propagator, where $\Delta_\ell$ is the group Laplacian acting on the $\ell$-th argument of the field. 
Indeed it just reflects the asymptotic quadratic decay ``$1/p^2$" of that propagator at large momentum $p$.

If the subgraph $\cH$ is the union of several connected components $\cH_k$, the divergence degree
factorizes as the sum of the divergence degrees of the connected components, from our very definition of connectedness
as rectangular block-factorization of the $\epsilon_{ef}$ incidence matrix:
\beq  \omega (\cH)  = \sum_k \omega (\cH_k)  .
\eeq

In the case of non-commutative TGFT's the ordering of faces results in a more subtle 
single-slice power-counting, established in \cite{valentinmatteo}. The factor $F-r$ is still multiplied by the dimension of the Lie group; 
but in the case of non-commutative groups and of graphs triangulating non-simply connected pseudo-manifolds,
the rank has to be supplemented by another term, not necessarily proportional to the dimension 
of the group, resulting in a \emph{twisted} divergence degree $\omega_t$ \cite{valentinmatteo}.

\subsubsection{Multi-scale power-counting}
  
Consider a graph $\cG$.  Consider the multi-scale decomposition  $\cA_\cG = \underset{\mu}{\sum}  \cA_{\cG, \mu}$.
The multislice power-counting is a bound that at fixed momentum attribution $\mu$
factorizes over all the $ \cG_{i}^{(k)}$ nodes of the  Gallavotti-Nicol\'o tree (hereafter GN tree).

 \begin{proposition}[Multi-scale fundamental bound]
  
There exists a constant $K$ such that the following bound holds:
\beq\label{fund}
 \vert \cA_{\cG, \mu}  \vert   \le K^{L(\cG)}  \prod_i \prod_{ k \in \llbracket 1 , k(i) \rrbracket } M^{\omega [  \cG_{i}^{(k)}]}
\eeq
\end{proposition}
\begin{proof}
In the ordinary case of an ordinary connected graph with $L$ lines and $V$ vertices, the slice propagator bound is 
\beq
\vert C_i ( x, y )  \vert \leq K M^{(D-2) i } \e^{- \delta M^{i} \vert x-y \vert }\,,
\eeq
and the divergence degree is $\omega = (D-2)L - D r$, where $r$, the rank of the $\epsilon_{ev}$ matrix, is $V-1$ since the graph is connected. 
The multi-scale aspect of the bound requires to optimize the rank effects at each scale. 
It follows conveniently from the compatibility between two trees:
the GN tree at fixed assignment $\mu$ and a real spanning tree $T_\mu$ made of lines of the graph $\cG$ which allows to 
organize in an optimal way the integration over the vertex positions \cite{VincentBook}. More precisely,
it is always possible to require $T_\mu$ to be a subtree when restricted to each GN node. This is because
the GN nodes form an inclusion forest; $T_\mu$ is chosen recursively from leaves towards the root $G$
of the GN tree. One first picks a spanning tree in a leaf of the GN tree, then contracts that leaf to a vertex and continue
until the end, obtaining a set of lines $T_\mu$. The inclusion structure of the GN tree implies that any such $T_\mu$
is a tree in the full graph whose restriction to each node is also a tree in that node. 
Then one can forget the useless decay factors associated to the lines not in $T_\mu$;
they cannot improve the final bound except at the level of the constant $K$ at best. 

But the reduced incidence matrix for $T_\mu$ is still $V-1$ by $V$ and has still rank $V-1$; there is therefore still ``root" vertices to
choose in the total graph $\cG$ in order to obtain minors of maximal rank. This is done again recursively, but this time in the reverse order.
One can pick an arbitrary root vertex $v_0$ in $\cG$. It determines in each node a unique ``local" root vertex $v_{ik}$, which is either
$v_0$ if it belongs to the node, or the starting vertex on the unique path of $T_\mu$ leading out of the node to the root $v_0$.

Then performing the ``canonical" change of variables of Jacobian 1 associated to $T_\mu$ and $V_0$ from leaves of $T_\mu$ to the root, 
we can integrate all positions of all vertices of $G$ save $V_0$ through the decay $$\prod_{e \in T_\mu} \e^{- \delta M^{i(e)} \vert x_e-y_e \vert }$$
and the result gives $ K^{V-1}\prod_i \prod_{ k \in \llbracket 1 , k(i) \rrbracket } M^{-D [V ( \cG_{i}^{(k)})-1]}$ as desired.

In the tensorial group field theory case we proceed in the same way to combine the bound and the optimization over the scales. 
First we collect all lines factors $M^{2i}$ and rewrite them as
$$\prod_i \prod_{k  \in \llbracket 1 , k(i) \rrbracket } M^{(d D -2) L [  \cG_{i}^{(k)}]}$$
through the usual trivial identities $M^i = \prod_{j=1}^i M$.
Then we integrate all $\theta$ variables in any face, optimizing along a tree in each face
as in \cite{tensor_4d}. This results in a factor 
$$ K^{L(\cG)}   \prod_i \prod_{ k \in \llbracket 1 , k(i) \rrbracket } M^{-d D L ( \cG_{i}^{(k)})  + DF( \cG_{i}^{(k)})  }$$
which combined with the first one gives 
$$ K^{L(\cG)}   \prod_i \prod_{ k \in \llbracket 1 , k(i) \rrbracket } M^{-2 L ( \cG_{i}^{(k)})  +  DF( \cG_{i}^{(k)})  } .$$

It remains to perform the $\lambda$ integrals, using the remaining decay, which is
\beq \prod_f \e^{- \delta M^{i(f)} \vert \sum_e \epsilon_{ef} \lambda_e \vert}, 
\label{facetree}\eeq
where $i(f)$ is the lowest 
scale in the face $f$. These integrals should give the rank contribution to $\omega$. 
But how to optimize this effect according to the scale attribution $\mu$?
By analogy with the previous case we should select a restricted set of faces $F_\mu$ 
such that the submatrix $\epsilon_{ef}$ with $f$ 
restricted to $F_\mu$ still has rank $r_{i,k}$ in each $ \cG_{i}^{(k)}$ node, and forget the decay factors from the other faces in \eqref{facetree}. 
This is the analog of selecting the former spanning tree $T_\mu$, and throwing the loop lines decays.

To select $F_\mu$, we start again from the leaves of the GN tree and proceed towards its root $\cG$.
We consider a leaf $\cH$ and select a first subset of faces such that the restricted submatrix $\epsilon_{ef}$ with $f$
and $e$ in $\cH$ has maximal rank; then we \emph{contract} $\cH$ and continue the procedure for the reduced
graph and the reduced GN tree, until the root is reached. At the end we obtain a particular set of faces $F_\mu$.

At each node $ \cG_{i}^{(k)}$ we have discarded the full incidence columns for internal faces 
which were combinations of other columns of that node. But because such faces were internal, these full columns have zeros 
outside the $ \cG_{i}^{(k)}$ block. Hence removing them cannot have any effect on the lower nodes rank. The conclusion is again
that the incidence matrix reduced to $F_\mu$, that is for which all internal faces not contained in $F_\mu$ have been discarded,
has still rank $r_{i,k}$ in each  $\cG_{i}^{(k)}$ node.

Discarding the decay factors for faces not in $F_\mu$, we now need to analyze the result of the integral 
\begin{equation} 
\int \prod_{e \in L(\cG)} d^D \lambda_e   \prod_{f \in F_\mu} \e^{- \delta M^{i(f)} \vert \sum_e \epsilon_{ef} \lambda_e \vert}, 
\end{equation}
and prove that it gives $\prod_i \prod_{ k \in \llbracket 1 , k(i) \rrbracket } M^{-D r_{i,k}}$.
This is the analog of the variables change and choice of the root vertices in each node. 
In the graph $\cG$ we can pick a set $L_\mu$ of  exactly $\vert F_\mu \vert $ \emph{lines} such that the (square)  $F_\mu$ by $L_\mu$ minor 
$\epsilon_{ef}$ has non-zero determinant. The $L_\mu$ by $F_\mu$ square incidence matrix 
$\epsilon_{ef}$ must still have exactly $r_{i,k}$ rank in each $\cG_{i}^{(k)}$ node (otherwise the $r_{i,k}$ columns $\epsilon_{ef}$ for $f \in F_\mu \cap \cG_{i}^{(k)}$
would not generate a space of dimension $r_{i,k}$, and the rank of the selected $F_\mu$ by $L_\mu$ square matrix would be strictly smaller than $F_\mu$).

We can now fix all values of the $\lambda_e$ parameters of the lines not in $L_\mu$ and consider the integrals 
\beq \int \prod_{e \in L_\mu} d^D \lambda_e \prod_{f \in F_\mu} \e^{- \delta M^{i(f)} \vert \sum_e \epsilon_{ef} \lambda_e \vert}, 
\eeq
We change of variables so that the integral becomes
\beq \int J \prod_{f \in F_\mu} d^D x_f  \e^{- \delta M^{i(f)} \vert  x_f - y_f \vert}, 
\eeq
where the $y_f$ variables are functions of the fixed $\lambda_e$ parameters of the lines not in $L_\mu$ and $J$ is a Jacobian.
This integral gives $\prod_{f \in F_\mu}  M^{-D i(f)} $, which by the condition on $F_\mu$ turns into 
$\prod_i \prod_{ k \in \llbracket 1 , k(i) \rrbracket } M^{-D r_{i,k}}$
as expected.

Remark that by Hadamard bound, since each column of this determinant is made of at most
$d$ factors $\pm 1$ (a line containing at most $d$ internal faces), the Jacobian $J$ of the corresponding change of variables 
is not very big, at most $\sqrt{d} ^{F_\mu}$, hence can be absorbed in the $K^{L(\cG)} $ factor. 

Finally we can integrate the fixed $\lambda_e$ parameters 
for $e \not \in L_\mu$ at a cost bounded by $K^{L (\cG)}$ if our group is compact. The case of a non-compact group 
requires infrared regularization and will not be treated here.

\end{proof}

%% file: sections/u1_model.tex
 

\section{Application to the $\U(1)$ 4d model.}

In this section, we illustrate the general formalism outlined in the previous sections, specializing to $d = 4$ and $G = \U(1)$. We will use angle coordinates $\theta_\ell \in \left[ 0 , 2 \pi \right[$ for the group elements 
$g_\ell = \e^{\rm{i} \theta_\ell}$, parameterizing the field $\vphi(\theta_1 , \dots , \theta_4)$.

These Abelian models will turn out to be very similar to scalar models with polynomial interactions in two-dimensional ordinary quantum field theory \cite{Simon}. The latter are super-renormalizable for any interaction, and the divergences can be subtracted by a simple change of variables at the level of the action. This procedure, called Wick ordering, removes the divergent tadpole contributions, yielding a perturbatively finite theory. 
In our TGFT context, local and polynomial interactions are replaced by finite sums of connected tensor invariants. Analogously to $P(\phi)_2$ models, we will prove that any such interaction generates a super-renormalizable model. We will then provide a generalization of the Wick ordering procedure, again yielding perturbatively finite models.  

  \subsection{Bound on the divergence degree}

Since $D = 1$, the divergence degree of a connected subgraph $\cH \subset \cG$ is given by
\beq
\omega (\cH) = - 2 L(\cH) + F(\cH) - r(\cH).
\eeq
We need to determine the set of divergent subgraphs, that is those $\cH$ such that $\omega (\cH) \geq 0$. In order to prove the model renormalizable, it will also be necessary to find a uniform decay of the amplitude associated to convergent graphs ($\omega < 1$), with respect to their external legs. In this respect, a suitable bound on $\omega$ in terms of simple combinatorial quantities will be sufficient. We can for instance decompose the number of faces with respect to the number of lines they consist of. We call $F_k$ the number of internal faces with $k$ lines, and $F_{ext,k}$ the number of external faces with $k$ lines, so that:
\beq
F = \sum_{k \geq 1} F_k\,, \qquad F_{ext} = \sum_{k \geq 0} F_{ext, k}\,.
\eeq
Note that, contrary to internal faces, external faces of $\cH$ do not have to contain a line of $\cH$, which explains that the second sum starts from $k=0$. We can also express the number of lines in terms of these quantities. Since $4$ different faces run through each line of $\cH$, we have:
\beq
4 L = \sum_{k \geq 1} k F_k + \sum_{k \geq 1} k F_{ext, k}\,,
\eeq
where in this formula both sums start with $k = 1$. We can therefore rewrite $\omega$ as
\beq
\omega = \sum_{k \geq 1} \left( 1 - \frac{k}{2} \right) F_k - \sum_{k \geq 1} \frac{k}{2} F_{ext, k} - r \,.
\eeq

We remark that the only positive contribution in this sum is given by $F_1$, to which only $p$-dipoles with $p \geq 2$ contribute. More precisely, 
\beq
F_1 = D_2 + 2 D_3 + 3 D_4 + 4 D_5\,,
\eeq
where $D_p$ is the number of $p$-dipole lines in $\cH$. We are thus lead to find a bound on $r$ in terms of these numbers of dipoles, which is the purpose of the following lemma.
\begin{lemma}
The rank of the incidence matrix associated to a connected graph $\cH$ verifies:
\beq
r \geq D_2 + D_3 + D_4 + D_5\,.
\eeq
\end{lemma}
\begin{proof}
Each $p$-dipole with $p \geq 2$ contains at least one internal face, which is independent of all the faces appearing in other lines. 
\end{proof}

Plugging this inequality into the expression of $\omega$ yields the following bound:
\beq\label{bound_om}
\omega \leq D_5 + \frac{D_4}{2} - \frac{D_2}{2} - \sum_{k \geq 3} \left( \frac{k}{2} - 1 \right) F_k - \sum_{k \geq 1} \frac{k}{2} F_{ext, k}\,.
\eeq
Note also that $D_5$ is always $0$, unless $\cH$ is the unique vacuum graph with a single line (sometimes called \textit{supermelon}). So the only non-trivial positive contribution comes from the $4$-dipoles. This seems to suggest that only melopoles will be convergent, which we confirm in the next section.

  \subsection{Divergent and convergent graphs}

To control the contribution of $D_4$ in (\ref{bound_om}), we take a step back and analyse the (exact) effect on $\omega$ of a $4$-dipole contraction in a connected graph $\cH$. Because $4$-dipoles are tracial, the question is well-posed.

\begin{proposition}
Let $\cH$ be a connected subgraph, and $l$ a $4$-dipole line. Then
\beq
\omega(\cH) = \omega(\cH / l)\,.
\eeq
\end{proposition}
\begin{proof}
We have immediately $L (\cH / l) = L(\cH) - 1$ and $F (\cH / l) = F(\cH) - 3$. As for the rank of the incidence matrix, it is easy to see that: $r(\cH / l) = r(\cH) - 1$. Therefore:
\beq
\omega(\cH / l) = \omega(\cH) + 2 - 3 + 1 = \omega( \cH )\,.
\eeq
\end{proof}

This property can be used to recursively reduce the analysis to that of graphs with a few melonic lines. For such graphs, (\ref{bound_om}) is constraining enough, and we can obtain the following classification.
\begin{proposition}\label{fund_u1}
Let $\cH \subset \cG$ be a connected subgraph.
\begin{itemize}
\item If $\omega(\cH) = 1$, then $\cH$ is a vacuum melopole.
\item If $\omega(\cH) = 0$, then $\cH$ is either a non-vacuum melopole, or a \textit{submelonic vacuum graph} (see Figure \ref{submelonic}). 
\item Otherwise, $\omega(\cH) \leq -1$ and $\omega(\cH) \leq - \frac{N(\cH)}{4}$.
\end{itemize}
\end{proposition}
\begin{proof}
Let us first assume that $\cH$ is a vacuum graph. We can perform a maximal set of successive $4$-dipole contractions, so as to obtain a graph $\widetilde{\cH}$ with $D_4 = 0$ and same power-counting as $\cH$. If $D_5 (\widetilde{\cH}) = 1$, then $\widetilde{\cH}$ is the supermelon graph, which means that $\cH$ is a melopole, and $\omega(\cH) = \omega(\widetilde{\cH}) = - 2 + 4 - 1 = 1$. On the other hand, when $D_5 (\widetilde{\cH}) = 0$, equation (\ref{bound_om}) gives
\beq
\omega( \widetilde{\cH} ) \leq  - \frac{D_2 (\widetilde{\cH}) }{2} - \sum_{k \geq 3} \left( \frac{k}{2} - 1 \right) F_k ( \widetilde{\cH} )\,,
\eeq
from which we infer that $\omega( \widetilde{\cH} ) \leq -1$ unless perhaps when $D_2 (\widetilde{\cH}) = F_k (\widetilde{\cH} ) = 0$ for any $k \geq 3$. But it is easy to see that these conditions immediately imply that $\widetilde{\cH}$ has one of the structures shown in Figure \ref{2cases}. A direct calculation then confirms that $\omega = 0$ for the left drawing, but $\omega = 1$ for the drawing on the right. This finally shows that $\omega = 0$ graphs are exactly the minimal graph on the left side of Figure \ref{2cases} dressed with additional melopoles, as shown in Figure \ref{submelonic}. We propose to call them \textit{submelonic vacuum graphs}. 

\begin{figure}
\begin{center}
\includegraphics[scale=0.5]{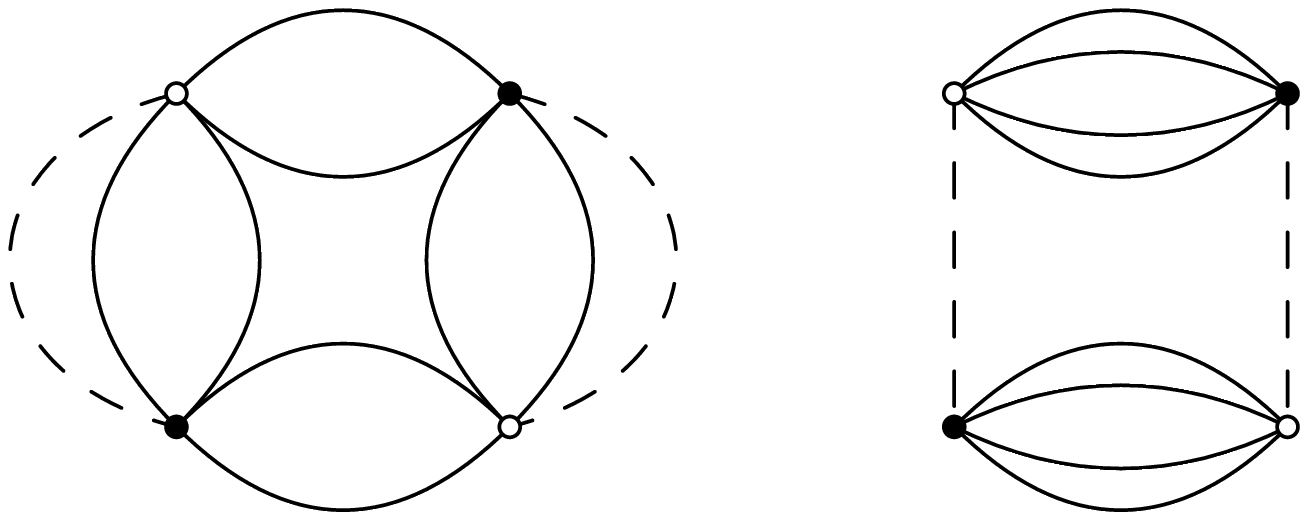}
\caption{Two vacuum graphs with $D_2 (\widetilde{\cH}) = F_k (\widetilde{\cH} ) = 0$ for any $k \geq 3$.}
\label{2cases}
\end{center}
\end{figure}

\
Let us now consider the case of a non-vacuum graph $\cH$. We can again perform a maximal set of $4$-dipole contractions 
and construct a new graph $\widetilde{\cH}$ verifying either: a) $L(\widetilde{\cH}) = D_4 (\widetilde{\cH}) = 1$; or b) $ D_4 (\widetilde{\cH}) = 0 $. In situation a), $\widetilde{\cH}$ reduces to a single $4$-dipole line, $\cH$ itself is a melopole, and $\omega(\cH) = \omega(\widetilde{\cH}) = - 2 + 3 - 1 = 0$. In situation b), the bound on $\omega$ gives
\beq
\omega( \widetilde{\cH} ) \leq  - \frac{D_2 (\widetilde{\cH}) }{2} - \sum_{k \geq 3} \left( \frac{k}{2} - 1 \right) F_k ( \widetilde{\cH} ) - \sum_{k \geq 1} \frac{k}{2} F_{ext , k} (\widetilde{\cH}) < 0\,,
\eeq
which shows that $\omega ( \widetilde{\cH} ) = \omega ( \cH ) \leq - 1$. We can finally prove a decay in terms of the number of external lines. For instance, we remark that the connectedness of $\widetilde{\cH}$ implies that at least one face going through a given external leg is of the type $F_{ext , k}$ with $k \geq 1$. And because each of these faces contains two external legs, we have $\sum_{k \geq 1} F_{ext , k} \geq \frac{N}{2}$. So we finally obtain
\beq
\omega(\cH) = \omega( \widetilde{\cH} ) \leq - \sum_{k \geq 1} \frac{k}{2} F_{ext , k} (\widetilde{\cH}) 
\leq - \frac{1}{2} \sum_{k \geq 1} F_{ext , k} \leq - \frac{N}{4}\,.
\eeq
 
\
All possible situations have been scanned, which ends the proof. 
\end{proof}

This classification allows to identify melopoles as the only source of divergences in the scale decomposition of non-vacuum connected amplitudes. Any model with a finite set of $4$-bubble interactions comes with a finite number of melopoles, and is therefore expected to be super-renormalizable. The purpose of the next sections is to prove that it is indeed the case.

\begin{figure}
\begin{center}
\includegraphics[scale=0.5]{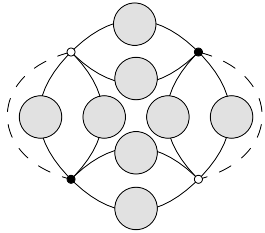}
\caption{The class of submelonic vacuum graphs: grey blobs represent melopole insertions.}
\label{submelonic}
\end{center}
\end{figure}

  \subsection{Melordering}

In the usual super-renormalizable $P(\phi)_2$ field theory \cite{Simon}, the finite set of counter-terms that are needed to tame divergences is simply provided by Wick ordering. It consists in a simple change of basis of interaction invariants, the coupling constants in this new basis being the renormalized ones. The net effect of Wick ordering at the level of the Feynman expansion is to simply cancel the contributions of graphs with tadpoles. This suggests a similar strategy to remove the special kind of tadpoles that are responsible for the divergences of our tensorial model, that is the melopoles. We will call this particular version of Wick ordering the \textit{melordering}.

\
Before going to the details of melordering, a few preliminary remarks are in order. As in the rest of this paper, we have to face a few subtleties introduced by the refined notion of connectedness on which TGFT relies.
In scalar theories, tadpole lines are exactly local objects, in the sense that their contributions can be factorized exactly. This is the reason why Wick ordering can be defined as a choice of a family of orthogonal polynomials with
respect to the regularized covariance. When such invariants are used as a basis to express the interaction part of the action, their expectation values in the vacuum is zero, and more generally all
tadpole contributions cancel out exactly. In tensorial theories however, we have seen that tadpoles can only be approximately local, at the condition of them being tracial (which melopoles are).
We therefore cannot hope to cancel them exactly, but only to eliminate their local divergent part. An important consequence is for example that melordered invariants will not necessarily have zero expectation value in the vacuum, but only a finite one (at the additional condition that submelonic vacuum counter-terms are added when needed, see section \ref{sec:sub}).

\
We now proceed with the definition of melordering.
Let us call $\Inv$ the vector space of connected tensor invariants, generated by the $4$-bubbles. Associated to the regularized covariance $C^{\rho}$, we want to define a linear and bijective map
$\Omega_{\rho}: \Inv \mapsto \Inv$ 
that maps any $4$-bubble to a suitably weighted sum of lower order $4$-bubbles. Getting inspiration from the scalar case, one should define $\Omega_{\rho}(I_{b})$ as a sum over pairings of the
external legs of $b$. The relevant pairings will be those resulting in one or several melopoles. Indeed, as we will see in an explicit example (see Appendix), a single connected invariant can give rise to several disconnected melopoles. For this reason, and despite the super-renormalizable nature of the model, the counter-terms have already a rich structure, only captured by the full machinery of Zimmermann forests. In such an approach, the renormalized amplitudes are given by sums over inclusion forests of divergent subgraphs $\cF$, of contractions of the bare amplitudes
\beq
\cA_\cG^R = \sum_\cF \prod_{\cH \in \cF} (- \tau_\cH ) \cA_\cG \,.
\eeq
In our case, the relevant structure is given by inclusion forests of melopoles, which we call \textit{meloforests} and define with respect to both subgraphs and bubble invariants.
\begin{definition}
\begin{enumerate}[(i)]
\item Let $\cH \in \cG$ be a subgraph. A meloforest $\cM$ of $\cH$ is a set of non-empty and connected melopoles of $\cH$, such that: for any $m , m' \in \cM$, either $m$ and  $m' $ are face and line-disjoint (i.e. have neither common lines nor common faces), 
or $m \subset m'$ or $m' \subset m$. We note $\cM(\cH)$ the set of meloforest of $\cH$.
\item Let $b$ be a $4$-bubble. A meloforest $\cM$ of $b$ is a meloforest for a graph made of a single vertex $b$. We call $I_{b, \rho}^{\cM}$ the observable associated to the smallest such graph, namely $\underset{m \in \cM}{\bigcup} m$. We note $\cM(b)$ the set of meloforests of $b$.
\end{enumerate}
\end{definition}
Meloforests have a relatively simple structure, due to a uniqueness property \cite{uncoloring,universality}. \begin{lemma}
Let $b$ be a $4$-bubble. There exists a unique vacuum graph $\cG$ with a single vertex $b$, such that any meloforest of $b$ is a meloforest of $\cG$.
\end{lemma}
\begin{proof}
As remarked in \cite{uncoloring,universality}, only melonic $2$-point subgraphs (in the sense of colored graphs) of $b$ can be closed in melopoles, and there is a unique way of doing so. Closing the maximal $2$-point subgraphs of $b$ in such a way results therefore in the unique graph $\cG$.
\end{proof}

We can now proceed with the definition of the melordering map.
\begin{definition}
For any $d$-bubble $b$, associated to the invariant $I_b$, and a cut-off $\rho$, we define the \textit{melordered invariant} $\Omega_{\rho}(I_{b})$ as
\beq
\Omega_{\rho}(I_{b}) \equiv \sum_{\cM \in \cM(b)} \prod_{m \in \cM} \left( - \tau_{m} \right) I_{b, \rho}^{\cM}\,.
\eeq
\end{definition}
By convention, the sum over meloforests includes the empty one, so that $\Omega_{\rho}(I_{b})$ as same order as $I_{b}$. Products of contraction operators are commutative, the definition is therefore unambiguous. These (non-trivial) products of contractions ensure that each term in the sum is a weighted $d$-bubble invariant, making $\Omega_{\rho}$ a well-defined linear map from $\Inv$ to itself. An example is worked out explicitly in the Appendix.

\
Consider now the theory defined in terms of melordered interaction at cut-off $\rho$, with partition function:
\bes
\cZ_{\Omega_\rho} &=& \int \extd \mu_{C_\rho} (\vphi , \vphib) \, \e^{- S_{\Omega_\rho}(\vphi , \vphib )} \,, \\
S_{\Omega_\rho}(\vphi , \vphib ) &=& \sum_{b \in \cB} t_b^R \, \Omega_{\rho}(I_b )(\vphi , \vphib).
\ees
We shall then consider the perturbative expansion in the renormalized couplings $t^R_b$ and prove that the corresponding Feynman amplitudes are finite. Let us call $\cS_N^{\Omega_\rho}$ the $N$-point Schwinger function of the melordered model. The next proposition shows that renormalized amplitudes have the expected form.
\begin{proposition}
The $N$-point Schwinger function $\cS_N^{\Omega_\rho}$ expands as:
\beq
\cS_N^{\Omega_\rho} = \sum_{\cG \; \mathrm{connected}, N(\cG)= N} \frac{1}{s(\cG)} \left(\prod_{b \in \cB} (- t_b^R )^{n_b (\cG)}\right) \cA_\cG^{R} \,,
\eeq
where the renormalized amplitudes can be expressed in terms of the bare ones as
\beq
\cA_\cG^R = \left( \sum_{\cM \in \cM(\cG)} \prod_{m \in \cM} \left( - \tau_{m} \right) 
\right) \cA_\cG \,.
\eeq
\end{proposition}
\begin{proof}
We first remark that the set $\cM (\cG)$ of meloforests of $\cG$ can be described according to meloforests of bubble vertices $b \in \cB(\cG)$:
\beq
\cM (\cG) = \left\{ \underset{b \in \cB}{\bigcup} \cM_b | \cM_b \; { \rm meloforest} \; {\rm of} \; b \in \cB(\cG) \right\}\,.
\eeq 
$\cA_\cG^R$ as defined above can therefore be written
\bes
\cA_\cG^R &=& \left( \sum_{ (\cM_b)_{b \in \cB(\cG)} } \prod_{b \in \cB(\cG)} \prod_{m \in \cM_b} \left( - \tau_{m} \right) 
\right) \cA_\cG \\
&=&  \prod_{b \in \cB(\cG)} \left( \sum_{\cM_b} \prod_{m \in \cM_b} \left( - \tau_{m} \right) 
\right) \cA_\cG \,.
\ees
Each element of the product over $b \in \cB(\cG)$ is a contraction operator taking all melopoles associated to $b$ into account. Let us fix a graph $\cG$ and a bubble $b$. Among the set of Wick contractions appearing in $\cS_N^{\Omega_\rho}$, the operator $\underset{\cM_b}{\sum} \underset{m \in \cM_b}{\prod} \left( - \tau_{m} \right)$ encodes all the terms due to the interaction $\Omega_\rho (I_b)$ that are compatible with the combinatorics of the external legs of $b$ in $\cG$ and the structure of the rest of the graph. We therefore understand that $\cS_N^{\Omega_\rho}$ as written above is a valid repackaging of all the Wick contractions generated by the melordered interaction.
\end{proof}
We will devote the whole section \ref{sec:finiteness} to proving that the renormalized amplitudes are indeed finite. Before that, we return to submelonic vacuum divergences.

  \subsection{Vacuum submelonic counter-terms}\label{sec:sub}

The melordering we just introduced is designed to remove melopole divergences, including logarithmic divergences of non-vacuum graphs and linear divergences resulting from vacuum melopoles. However, we have seen that a third source of divergences is given by submelonic vacuum graphs. They again concern tadpole graphs, so they can also be removed by adding extra counter-terms to the melordering of some of the bubbles. As long as we are concerned with computations of transition amplitudes, they are irrelevant since they will only affect $\cZ$ and none of the connected Schwinger functions.  But we include them here for completeness.  
 
\
We can define an \textit{extended melordering} $\overline{\Omega}_\rho$ that coincides with $\Omega_\rho$ for bubbles which cannot be closed in a submelonic vacuum graph, and adds additional counter-terms to those which can. We can call the latter \textit{submelonic bubbles}. They are exactly the bubbles that reduce to a four-point graph as in Figure \ref{sub_open} once all the melonic parts have been closed into melopoles and contracted. Such bubbles generate additional divergent forests, which we can call \textit{submelonic forests}:
\begin{definition}
Let $b$ be a submelonic bubble. A submelonic forest of $b$ is a forest $\cS = \cM \cup \{ \cG \}$, where $\cM$ is a melonic forest and $\cG$ is a vacuum graph with a single vertex $b$. We call $I_{b, \rho}^{\cS}$ the amplitude associated to the graph $\cG$. We call $\cS(b)$ the set of submelonic forests of $b$.
\end{definition}   
{\bf Remark.} Given a submelonic bubble, there are exactly two possible choices for $\cG$, which correspond to the two possible ways of closing the melopole-free graph of Figure \ref{sub_open}.

\begin{figure}
\begin{center}
\includegraphics[scale=0.5]{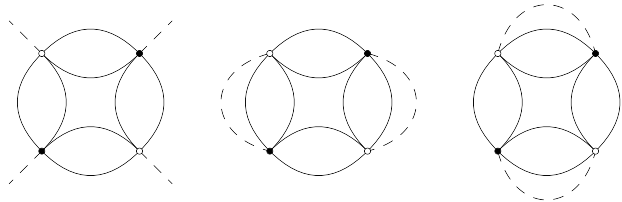}
\caption{On the left: structure of a submelonic bubble once all melonic parts have been closed into melopoles and contracted. On the right: the two ways of obtaining a submelonic vacuum graph.}
\label{sub_open}
\end{center}
\end{figure}

The extended melordering is finally defined by
\beq
\overline{\Omega}_{\rho}(I_{b}) \equiv \sum_{\cM \in \cM(b)} \prod_{m \in \cM} \left( - \tau_{m} \right) I_{b, \rho}^{\cM} + \sum_{\cS \in \cS(b)} \prod_{s \in \cS} \left( - \tau_{s} \right) I_{b, \rho}^{\cS}
\eeq
when $b$ is submelonic. This implies similar formulas for renormalized amplitudes in the extended melordered model, which in particular do not affect the expression for melordered connected Schwinger functions. The only difference will be that the partition function of the extended melordered model will be well-defined as a formal series, contrary to the simple melordering for which $\cZ$ will have some logarithmically divergent coefficients.

%% file: sections/finiteness.tex
 

\section{Finiteness of the renormalized series}\label{sec:finiteness}

In this section, we will prove that melordered models with maximal interaction order $p < + \infty$ are perturbatively finite at any order. To avoid dealing with submelonic vacuum graphs, we will only focus on the connected Schwinger functions, which are the physically meaningful quantities after all. They are well-defined formal series in the renormalized couplings if all non-vacuum and connected renormalized amplitudes $\cA_\cG^{R}$ are finite.

\
We will again rely on the multi-scale analysis, following the usual procedure of \cite{VincentBook}, which consists in two steps. We first need to show that renormalized amplitudes associated to bare divergent graphs verify multi-scale convergent bounds. This is most conveniently done through a classification of divergent forests (in our case meloforests), which for a given scale attribution $\mu$, splits these in two families: the \textit{dangerous} ones, associated to high subgraphs, that cancel genuine divergences; and on the other hand \textit{inoffensive} divergent forests that do not have any quasi-locality property, henceforth do not serve any purpose. The inoffensive forests bring finite contributions that do not ruin the power-counting, and can rather be interpreted as a drawback of the renormalized series: they have no physically meaningful consequence, and in addition (in just renormalizable models, but not for super-renormalizable models like the ones treated here) results in ``renormalon effects" that typically prevents from constructing a convergent series. In a second step, we will prove that the sum over scale attributions can be performed, and the cut-off $\rho$ sent to infinity while keeping the amplitudes finite. 

  \subsection{Classification of forests}

We follow the general classification procedure of \cite{VincentBook}, that at each scale attribution allows to factorize the contraction operators defining the renormalized amplitude. Let $\cG$ be a connected (non-vacuum) graph. We can decompose the renormalized amplitude $\cA_\cG^R$ in terms of its scale attributions:
\beq\label{2sums}
\cA_\cG^R = \sum_\mu \sum_{\cM \in \cM(\cG)} \prod_{m \in \cM}(- \tau_m) \cA_{\cG , \mu} \,. 
\eeq
The classification of forests is a reshuffling of the sum over meloforests that allows to permute the two sums. We know that for a given scale attribution $\mu$, the forests that contribute to the divergences are those containing high melopoles. We therefore need to define the notion of high meloforest, and reorganize the sum in terms of these quantities. We follow the standard procedure, and start with the following set of definitions.

\begin{definition}
Let $\cG$ be a connected graph, $\mu$ a scale attribution, and $\cM$ a meloforest of $\cG$.
\begin{enumerate}[(i)]
\item We say that a subgraph $g \subset \cG$ is \textit{compatible} with a meloforest $\cM$ if $\cM \cup \{ g \}$ is a forest.
\item If $g$ is compatible with a meloforest $\cM$, we note $B_{\cM} (m)$ the \textit{ancestor} of $g$ in $\cM \cup \{ g \}$, and we similarly call $A_\cM (g) \equiv \{m \subset g | m \in \cM \}$ the \textit{descendants}.
\item \textit{Internal and external scales} of a compatible graph $g$ in a meloforest $\cM$ are defined by:
\beq
i_{g , \cM}(\mu) = \inf_{e \in L( g \backslash A_\cM (g) )} i_e(\mu)\,, 
\qquad e_{g , \cM}(\mu) = \sup_{e \in N(g) \cap B_\cM (g)} i_e(\mu).
\eeq
\item The \textit{dangerous} part of a meloforest $\cM$ with respect to $\mu$ is:
\beq
D_\mu (\cM) = \{ m \in \cM | i_{m , \cM}(\mu) > e_{m , \cM}(\mu) \}\,,
\eeq 
and the \textit{inoffensive} part is the complement $I_\mu (\cM) = \cM \backslash D_\mu (\cM)$. Finally $I(\mu)$ is the set of all inoffensive forests in $\cG$.
\end{enumerate}
\end{definition} 
{\bf{Remarks.}} The notions of internal and external scales with respect to a meloforest are consistent with the previous definitions, since $i_{m , \emptyset} = i_m$ and $e_{m, \emptyset} = e_m$. Moreover, (non-vacuum) melopoles have exactly two external legs, which makes the situation relatively simple.

The following important lemma leads to the partition of forests.
\begin{lemma}
Given a meloforest $\cM$,
\beq
I_\mu ( I_\mu (\cM) ) =  I_\mu (\cM)\,. 
\eeq
\end{lemma}
\begin{proof}
Similar to \cite{VincentBook}, but simpler. 
\end{proof}

This implies that the set of meloforests $\cM( \cG )$ of a connected graph $\cG$ can
be partitioned, according to the inoffensive forests associated to any scale attribution $\mu$:
\beq
\cM( \cG ) = \underset{\cM | I_\mu (\cM) = \cM }{\bigcup} \{ \cM' | I_\mu ( \cM' ) = \cM \}\,.
\eeq
We can finally characterise the equivalence class of a meloforest $\cM$ by introducing its maximal forest $\cM \cup H_\mu (\cM)$, where
\beq
H_\mu (\cM) = \{ m \; {\rm compatible} \; {\rm with} \; \cM \, |  \; m \in D_\mu ( \cM \cup \{ m \})\}\,.
\eeq
We can indeed show that:
\begin{proposition}
For any $\cM \in I(\mu)$, $\cM \cup H_\mu (\cM)$ is a meloforest, and moreover:
\beq
\forall \cM' \in \cM(\cG), \, I_\mu(\cM') = \cM \Longleftrightarrow \cM \subset \cM' \subset \cM \cup H_\mu (\cM)\,.
\eeq
\end{proposition}
\begin{proof}
Similar to \cite{VincentBook}, but simpler. 
\end{proof}

This finally allows to reorganize the operator defining the renormalized amplitude as
\beq
\sum_{\cM \in \cM(\cG)} \prod_{m \in \cM}(- \tau_m) = \sum_{\cM \in I(\mu)} \prod_{m \in \cM} (- \tau_m) \prod_{h \in H_\mu (\cM)} (1 - \tau_h)\,,
\eeq
which decomposes the product of contraction operators into inoffensive parts and high parts. And since it holds for any $\mu$, we can use this formula to invert the two sums in (\ref{2sums}) and obtain:
\bes\label{classified}
\cA_\cG^R &=& \sum_{\cM \in \cM(\cG)} \cA_{\cG , \cM}^R\,,\\
\cA_{\cG , \cM}^R &=& \sum_{\mu | \cM \in I(\mu)} \prod_{m \in \cM} (- \tau_m) \prod_{h \in H_\mu (\cM)} (1 - \tau_h) \cA_{\cG , \mu}\,.
\ees
The factorization (\ref{classified}) is key to the proof of finiteness. We shall first show that, with respect to the bare theory, the power-counting of $\cA_{\cG , \cM}^{R}$ for a given scale attribution is improved, and is always convergent. We will then explain why the sum over scales is finite given such convergent multi-scale bounds. The final sum over meloforests will not bring more divergences, since their cardinal is finite (and even bounded by $K^{n(\cG)}$
for some $K>0$).
	
  \subsection{Power-counting of renormalized amplitudes}	

Let us fix a meloforest $\cM$ and a scale attribution $\mu$ such that $\cM \in I(\mu)$. The product of operators acting on $\cA_{\cG , \mu}$ in (\ref{classified}) can be computed explicitly. We can for example first act with $\underset{m \in \cM}{\prod} \tau_m$ which evaluates as
\beq
\prod_{m \in \cM} \tau_m \, \cA_{\cG , \mu} = \left( \prod_{m \in \cM} \nu_\mu (m / {A_\cM(m)}) \right) \cA_{\cG / \cM , \mu}\,,
\eeq
where $\cG / \cM$ is the graph obtained from $\cG$ once all the subgraphs of $\cM$ have been contracted. This graph is nothing but $\cG / {A_{\cM}(\cG)}$. 
$\nu_\mu$ is a generalized notion of amplitude associated to subgraphs $\cH \subset \cG$, that just discards the contributions of external faces. In this sense, it is analogue to an amputated amplitude in usual field theories. In particular, we can assume that $\cA_{\cG / \cM , \mu}$ is an amputated amplitude and write:
\beq\label{intermediate}
\prod_{m \in \cM} \tau_m \, \cA_{\cG , \mu} = \prod_{ g \in \cM \cup \{ \cG \}} \nu_\mu (g / {A_{\cM }(g)})\,.
\eeq
The power-counting, which only depends on internal faces, is unaffected by the fact that we are working with such amputated amplitudes, and we conclude that
\beq
\vert \prod_{m \in \cM} (- \tau_m) \cA_{\cG , \mu} \vert \leq K^{L(\cG)} \prod_{ g \in \cM \cup \{ \cG \}} \prod_{(i , k)} M^{\omega[ ( g / {A_{\cM}(g)} )_i^{(k)} ]}\,.
\eeq 
This is a generalization of the power-counting (\ref{fund}), and reduces to it when $\cM = \emptyset$. This proves that the sum over inoffensive forests does not improve nor worsen the power-counting, as was expected. Finiteness is entirely implemented by the useful part of the contraction operators, namely $\underset{h \in H_\mu (\cM)}{\prod} (1 - \tau_h)$. To make this apparent, we first write it as
\beq
\prod_{h \in H_\mu (\cM)} (1 - \tau_h) = \prod_{g \in \cM \cup \{ \cG \}} \prod_{h \in H_\mu (\cM) | B_\cM(h) = g} (1 - \tau_h)
\eeq
and act on (\ref{intermediate}) to get
\beq\label{intermediate2}
\vert \prod_{h \in H_\mu (\cM)} (1 - \tau_h) \prod_{m \in \cM} (- \tau_m) \cA_{\cG , \mu}  \vert
=  \prod_{ g \in \cM \cup \{ \cG \}} \prod_{h \in H_\mu (\cM) | B_\cM(h) = g} \vert (1 - \tau_h) \,  \nu_\mu (g / {A_{\cM}(g)}) \vert\,.
\eeq
Now, the effect of each $(1 - \tau_h)$ is to interpolate one of the variables of (at most two) external propagators in $N (h) \cap ( g / {A_{\cM}(g)} )$. For example, assuming the fourth variable is concerned (that is $h$ is a melopole that has been inserted on a colored line of color $4$), we have something of the form
\beq
C_{i}( \theta_1 , \dots , \theta_4 ; \theta_\ell') - C_i( \theta_1 , \dots , \tilde{\theta}_4 ; \theta_\ell')
= \int_{0}^{1} \extd t \left( \theta_4 - \tilde{\theta}_4  \right) \frac{\partial}{ \partial \theta_4} C_i( \theta_1 , \dots , \tilde{\theta}_4 + t (\theta_4 - \tilde{\theta}_4) ; \theta_\ell')\,,
\eeq
with $i \leq e_{\cM , h} (\mu)$. Moreover, since $h$ is high in $g / {A_{\cM}(g)}$, $| \theta_4 - \tilde{\theta}_4 |$ is at most of order $M^{- i_{\cM , h} (\mu)} $. So using the bound (\ref{deriv_bound}) on derivatives of the propagator, we conclude that $(1 - \tau_h)$ improves the bare power-counting by a factor:
\beq
M^{i} |\theta_4 - \tilde{\theta}_4| \leq K M^{ e_{\cM , h} (\mu) - i_{\cM , h} (\mu) }\,.
\eeq

This additional decay allows to prove the following proposition.
\begin{proposition}\label{pc_r}
There exists a constant $K$ such that for any graph $\cG$ and meloforest $\cM$:
\beq
| \cA_{\cG, \cM}^{R} | \leq  K^{L(\cG)}  \sum_{\mu | \cM \in I(\mu)} \prod_{ g \in \cM \cup \{ \cG \}} \prod_{ (i , k) } M^{\omega'[ ( g / {A_{\cM}(g)} )_i^{(k)} ]}\,,
\eeq
where
\beq
\omega'[ ( g / {A_{\cM}(g)} )_i^{(k)} ] = \min \{ - 1 , \omega[ ( g / {A_{\cM}(g)} )_i^{(k)} ] \}\,,
\eeq
except if $g \in \cM$ and $( g / {A_{\cM}(g)} )_i^{(k)} = g / {A_{\cM}(g)}$, in which case $\omega(( g / {A_{\cM}(g)} )_i^{(k)}) = 0$.
\end{proposition}
\begin{proof}
From (\ref{intermediate2}), and using the additional decays from operators $(1 - \tau_h)$, one improves the degree by a factor $-1$ for most of the high subgraphs. More precisely, this is possible for any high subgraph that has external legs in a contraction $g / {A_{\cM}(g)}$, that is any high subgraph 
$( g / {A_{\cM}(g)} )_i^{(k)}$ different from a root  $g / {A_{\cM}(g)}$.
\end{proof}

  \subsection{Sum over scale attributions}

Equipped with this improved power-counting, we can finally prove that the renormalized amplitudes are finite. For clarity of the presentation, let us first show it for a fully convergent graph $\cG$, that is a graph with no melopole.   
In this case, we know that:
\beq
| \cA_{\cG , \mu} | \leq K^{L(\cG)}  \prod_{ (i , k ) } M^{- N (\cG_i^{(k)}) / 4}\,,
\eeq
from which we need to extract enough decay in $\mu$ to sum over the scale attributions. Let $\cB(\cG)$ be the set of vertices (i.e. $4$-bubbles) of $\cG$, and for $b \in \cB(\cG)$ let us call $L_b (\cG)$ the set of lines that are hooked to it. We can define notions of internal and external scales associated to a bubble $b$:
\beq
i_b (\mu) = \sup_{l \in L_b (\cG)} i_l (\mu)  \,, \qquad  e_b (\mu) = \inf_{l \in L_b (\cG)} i_l (\mu)\,.
\eeq
We then remark that for any $i \in \mathbb{N}$ and $b \in \cB(\cG)$, $b$ touches a high subgraph $\cG_i^{(k)}$ if and only if $i \leq i_b (\mu)$. Moreover when it does, the number of high subgraphs $\cG_i^{(k)}$ that touch $b$ is certainly bounded by its number of external legs, and therefore by $p$. Hence we can assign a fraction $1 / p$ of the decay of a bubble to every high subgraph with respect to which it is external. This yields  
\beq
 \prod_{ (i , k ) } M^{- N (\cG_i^{(k)}) / 4} \leq  \prod_{ (i , k ) } \prod_{b \in \cB(\cG_i^{(k)})| e_b (\mu) < i \leq i_b (\mu)}
 M^{-  \frac{1}{4p}}\,,
\eeq 
by using the fact that $b$ is an external vertex of $\cG_i^{(k)}$ exactly when $e_b (\mu) < i \leq i_b (\mu)$. We can then invert the two products and obtain
\beq
| \cA_{\cG , \mu} | \leq K^{L(\cG)}  \prod_{ b \in \cB(\cG) } \prod_{ (i , k) | e_b (\mu) < i \leq i_b (\mu)} M^{-  \frac{1}{4p}} = K^{L(\cG)}  \prod_{ b \in \cB(\cG) } M^{- \frac{i_b (\mu) - e_b (\mu)}{3 p}}\,.
\eeq
Finally, since the number of pairs of legs hooked to a given vertex $b$ is bounded by $p (p -1) / 2$, we can finally conclude that
\beq
| \cA_{\cG , \mu} | \leq K^{L(\cG)}  \prod_{ b \in \cB(\cG) } \prod_{(l , l') \in L_b (\cG) \times L_b (\cG)} M^{- \frac{ 2 | i_{l'} (\mu) - i_l (\mu) | }{3 p^2 (p - 1)}}\,.
\eeq
With this decay at hand, the sum over scales can be performed by picking a 'tree of scales', very similarly to the choice of a tree adapted to the GN tree that establishes the power-counting. We refer to \cite{VincentBook} for more details, and just state the resulting proposition.

\begin{proposition}
There exists a constant $K > 0$ such that the amplitude of any fully convergent graph $\cG$ is absolutely convergent with respect to $\mu$, and moreover
\beq
\sum_{\mu} | \cA_{\cG , \mu} | \leq K^{L (\cG)} \,.
\eeq
\end{proposition}

We now explain why similarly, when $\cG$ contains melopoles, the sum over $\mu$ in (\ref{classified}) can be performed without cut-off. From the power-counting (\ref{pc_r}), and given that melopoles have at most two external legs, one notices 
that
\beq
\omega[( g / {A_{\cM}(g)} )_i^{(k)}] \leq - \frac{N (( g / {A_{\cM}(g)} )_i^{(k)}) }{2}\,.
\eeq
So the decay that was proven for convergent graphs (\ref{fund_u1}) generalizes to
\beq
| \cA_{\cG, \cM}^{R} | \leq  K^{L(\cG)}  \sum_{\mu | \cM \in I(\mu)} \prod_{ g \in \cM \cup \{ \cG \}} \prod_{ (i , k) } M^{- \frac{N ( ( g / {A_{\cM}(g)} )_i^{(k)} )}{4} }\,.
\eeq  
The strategy used for proving convergence of fully convergent graphs is therefore applicable. We conclude that $\cA_{\cG , \cM }$ is an absolutely convergent series in $\mu$, and even bounded by $K^{L(\cG)}$ for some constant $K$. The final sum over meloforests is not problematic, as the number of melopoles associated to a bubble $b$ is clearly bounded by a constant (for example $2^{p/2}$). This means that the number of meloforests associated to a vertex $b$ is also bounded by a constant $K_1 > 0$, and since meloforests of graphs are by definition unions of meloforests associated to single vertices, the number of meloforests of $\cG$ is itself bounded by $K_1 ^{n (\cG)}$. Overall, we conclude that:
\begin{proposition}
There exists a constant $K > 0$, such that the renormalized amplitude of any (non-vacuum) graph $\cG$ verifies:
\beq
| \cA_{\cG}^{R} | \leq K^{L( \cG )} \,.
\eeq
\end{proposition}
This not only proves renormalizability of the model, but also that there is no ``renormalon effect". The latter is a specific feature of our super-renormalizable model, that would not hold for more complicated just-renormalizable models. In such situations, it will be preferable to resort to the \textit{effective series}, because it is the unphysical sum over inoffensive forests automatically generated in the renormalized series that is responsible for this undesirable effect (see \cite{VincentBook}). 

\
We finally state the main theorem that was proven in this section.
\begin{theorem}
The melordered $U(1)$ model in $d = 4$, with an arbitrary finite set of $4$-bubble interactions, is perturbatively finite at any order.
\end{theorem}

Note that this finiteness theorem would still hold true had we relied on the usual notion of Wick ordering
, for the difference between Wick ordering and melordering is a sum of convergent terms. However, while the Wick ordering preserves locality in the usual spacetime sense, it is incompatible with the TGFT locality principle (i.e. tensor invariance), in two respects: first, a tadpole graph can only approximately factorize as a coefficient times a tensor invariant interaction; second, such an approximate factorization only occurs for particular tadpoles, the tracial ones. These are the reasons why: a) the sum over pairings of external legs defining the usual Wick ordering needs to be restricted to pairings yielding tracial graphs; b) each such pairing must be supplemented with a contraction operator, which extracts a tensor invariant contribution.   
This is what makes the usual Wick ordering inappropriate, and the introduction of the melordering necessary. In this manner, only tensor invariant counter-terms are introduced in the action, which approximate the divergent contributions of the melopoles. The fact that the latter are tracial is key to the whole construction.

%% file: sections/conclusion.tex
 

\section{Conclusion and outlook}
Let us summarize what we have achieved in this paper. 

First of all, we have set up the general framework for the multi-scale analysis of TGFT models. Multi-scale analysis had already been applied to some simpler TGFT models \cite{tensor_4d, josephsamary}, but the models considered in these recent works lacked one ingredient that requires a more refined analysis, as we have seen: the gauge invariance (closure) condition imposed on the TGFT field. In simplicial geometric models where the TGFT field is expected to describe geometric $(d-1)$-simplices, with its arguments representing normal vectors to its $(d-2)$-faces (or their conjugate gravitational connections), this condition imposes the closure of the same faces. The Abelian models considered in this paper do not have such a simplicial geometric interpretation. But this same condition introduces a local gauge invariance at the level of each $d$-cell of the cellular complex dual to the TGFT Feynman diagram, and a corresponding discrete gauge connection associated to each line of the same diagram. The Feynman amplitudes then take the form (and interpretation) of lattice gauge theories with gauge group given by the domain space of (each argument of the) TGFT field. Beside this interpretation, the same condition introduces an additional coupling between lines and faces of the TGFT Feynman diagram, and a more interesting dependence of the amplitudes on the topology of the diagram. 

In turn, we have seen that this more involved structure forces a revision and a generalization of some important notions of standard field theory. We provided such generalized notions, more precisely we introduced a new notion of {\it connectedness}, a new notion of {\it locality} (which we named {\it traciality}) and, stemming from them, a new procedure for {\it contraction of high subgraphs}. From the lattice gauge theory perspective on the amplitudes, this amounts to a new coarse graining procedure, that itself deserves to be investigated in more detail.

Among the new notions we introduce also that of {\it melonic Wick ordering}, or {\it melordering}, generalizing again the usual notion of Wick ordering of interaction monomials and $N$-point functions. Such ordering is a first step in standard renormalization of field theories, which removes divergences associated to tadpoles.

Armed with the multi-scale framework and these new notions, we then analyzed in detail a concrete Abelian TGFT model, which corresponds to the $U(1)$ type of model studied in \cite{tensor_4d} but with the additional gauge invariance condition turning it into a quantization of a $U(1)$ gauge theory. We prove that it is {\it super-renormalizable} for any choice of polynomial interaction. The Feynman amplitudes are convergent, with the exception of some TGFT analogues of tadpoles, which we call {\it melopoles}. The melordering procedure then removes these divergences and leaves us with a finite renormalized model. Notice the striking difference 
with the power-counting of standard local quantum field theory, which for arbitrary polynomial local interactions
is super-renormalizable in $d=2$ rather than $d=4$.
\

We now turn to an outlook on future developments. Having set up the general framework for multi-scale analysis of renormalizability of TGFTs, the natural thing to do is to tackle more elaborate TGFT models. Remaining within the same class of Abelian models we studied in this paper, with Laplacian kinetic term and invariant tensor interactions, the generalization of our analysis to higher dimensions and higher-dimensional Abelian groups requires more accurate bounds than those of section 5.1
and the proof of contractiblity and traciality to more general graphs than the melopoles. We have checked some examples which indicate that this extension 
should be doable.
We know already the general power-counting, established first in \cite{lin}, and the corresponding (single slice) divergence degree. It suggests that in $d=3$, a model with an Abelian gauge group of dimension 3 would be renormalizable up to polynomial interactions of order 6. On the other hand, it would also suggest that in $d=4$, only models based on groups of dimension 2 at most would be renormalizable.

This preliminary estimate gives hope that models more closely related to $3d$ quantum gravity, thus based on the group $SU(2)$ would be renormalizable
with polynomial invariant interactions at order up to 6. The complete analysis would be more complicated than the one we have performed, because of the non-Abelian nature of the gauge group, but we can already make some informed guesses on its outcome. Previous studies on power-counting and scaling bounds \cite{GFTrenorm,valentinmatteo, Magnen:2009at} in topological models (thus initially without a Laplacian in the kinetic term) suggest that the relevant multi-scale fundamental bound we have given in the Abelian case still holds, but now with each $\omega( \cG_{i}^{(k)})$ replaced by $\omega_t ( \cG_{i}^{(k)})$, that is a {\it twisted} divergence degree taking into account the ordering of lines in the boundary of faces of the graph, and dependent on the second twisted Betti number of the complex corresponding to it. This is also in accordance with our understanding of the (translation) symmetry of the corresponding lattice gauge theory and simplicial gravity path integral \cite{GFTdiffeos,valentinmatteo}. On the other hand, we may also expect the same twist to be absent for melonic graphs, since they triangulate the 3-sphere. If this is true, then the $SU(2)$ would be just renormalizable up to order 6 in the interaction, as the Abelian counterpart $U(1)^3$. Only a detailed analysis can give support to this expectation. 

\

One issue that should be tackled in trying to extend the analysis performed in this paper to models more directly related to $3d$ gravity and $BF$ theories (as a step towards proper quantum gravity models) is how the renormalizability is affected by the introduction of further gauge invariance projections within the interaction terms of the model we have studied. Indeed, one way to understand the type of invariant interactions we have used is that they arise naturally when integrating out $d$ of the $d+1$ colored fields in a colored TGFT with standard $d$-simplex interaction. However, when this integration is performed in a topological colored TGFT model in which gauge invariance is imposed on all fields entering the $d$-simplex interaction, the corresponding projector ends up attached to the internal colored lines in each interaction vertex of the resulting single-field model. The effect of these additional projections on the power-counting should then be studied carefully.

Other variations of the class of models that we studied in this paper, that may be worth investigating as well, are models with the same type of interactions but different kinetic terms. The choice of the Laplacian in the Euclidean case is suggested by analogy with standard field theory and by other considerations such as its reflection (Osterwalder-Schrader) positivity, but we do not have at present a complete axiomatic formulation of TGFTs that would select it as the only reasonable choice on physical grounds; therefore, other possibilities can be considered.

\ 

The real goal, however, of studying renormalizability of TGFT models is to tackle and understand TGFT formulations of 4d quantum gravity \cite{DPFKR,EPRL,BO-BC,BO-Immirzi}. This is at once challenging and very interesting. First of all, the same general issues pointed out above apply to these models as well (role of projectors, choice of kinetic terms, etc). Second, these models are obviously based on non-Abelian and, in the Lorentzian context, non-compact groups, with their additional complications and subtleties that our analysis did not deal with. Third, even our limited expectations for how the non-Abelian nature of the group affects the power-counting, based on our analysis, have to be taken with great care due to the specific construction of these models. Indeed, from the group-theoretic point of view, the main ingredient that gives 4d gravity models starting from topological $BF$ ones is the so-called {\it simplicity} constraint which amounts to restricting the domain space of the TGFT field to submanifolds of the 4d rotation or Lorentz group. This submanifold is 3-dimensional. However, it cannot be assimilated to a 3-dimensional group of the type we have dealt with in this paper or their non-Abelian version, because it is either an homogeneous space, for Barrett-Crane-like models \cite{DPFKR,BO-BC}, or just a 3d submanifold of the rotation or the Lorentz group in models involving an Immirzi parameter \cite{EPRL,BO-Immirzi}. Therefore, it is 
premature to guess at this stage what the status of such models could be, concerning renormalizability. Only a careful analysis will tell.

\

Last, we would like to mention the need to go beyond perturbative renormalizability. On this road a first step should be to build fully at the constructive level the models defined and perturbatively renormalized in this paper. This should be possible for any positive even monomial interaction, starting with the simplest case, namely the $\phi^4$ interaction. It should probably also work
for any polynomial semi-bounded interaction. Such a constructive analysis should 
prove that their Schwinger functions are the Borel sums\footnote{More precisely, the Borel-LeRoy sums of appropriate order for monomial interaction of degree higher than $4$.}
of their perturbative expansion. We are quite confident that this can be achieved using the technique of the \emph{loop vertex expansion} or LVE \cite{Rivasseau:2007fr} combined with a ``cleaning expansion" and non-perturbative bounds ``\`a la Nelson". Indeed a similar program was recently achieved in the case of the ordinary $\phi^4_2$ model \cite{Rivasseau:2011df}Ê and of 
the non-commutative super-renormalizable Grosse-Wulkenhaar model in two dimensions \cite{ZW2}
whose power-counting and positivity properties are comparable; furthermore we know that the LVE, 
created to tackle non-perturbatively matrix models, applies quite naturally 
also to tensor models \cite{Magnen:2009at, universality}. In the long term, this constructive perspective, currently lacking in other 
approaches to quantum gravity, is certainly a major asset of the TGFT approach.

The next steps concern the perturbative study of the renormalization group flows of more advanced renormalizable TGFT models and ultimately their non-perturbative construction. The study of the flows starts with computing explicitly the leading terms of their beta functions. Obviously, this is technically challenging. The impressive analysis of \cite{josephaf} for the simpler model of \cite{tensor_4d} seems to indicate that TGFTs could be generically asymptotically free. If true, this would be very important for TGFTs in general and for gravitational models in particular. A growing coupling constant and a most likely finite domain of analyticity of the TGFT free energy would imply that these models would dynamically (and thus somehow inevitably) undergo a phase transition. In turn, for models where a pre-geometric interpretation of the variables and amplitudes is possible (e.g. the 4d gravity models) in terms at least of simplicial geometry, such phase transition may be \cite{GFTfluid,lorenzoGFT,vincentTensor}, like in matrix models \cite{MM}, the hallmark of the continuum geometric limit of the models (geometrogenesis), which would then have to be studied in great detail to understand its physical implications and generic features.

%% file: sections/akno.tex

\section*{Acknowledgements}
It is a pleasure to thank Joseph Ben Geloun, Valentin Bonzom and Razvan Gurau for fruitful discussions in the early stages of this work, as well as for hosting us at the Perimeter Institute. 
This work is partially supported by a Sofja Kovalevskaja Award by the A. von Humboldt Stiftung, which is gratefully acknowledged. 
S.C. aknowledges travel funding from the European Union Seventh Framework Programme [FP7-People-2010-IRSES] under grant agreement number 269217.

%% file: sections/appendices/phi_6.tex
 

\section*{Appendix: Wick-ordering of a $\vphi^6$ interaction}

To illustrate the general statements of the paper, we give some more details for a model with a single $\vphi^6$ interaction:
\bes
S(\vphi , \vphib) &=& \int [\extd g_i]^{12} \vphi( g_{1} , g_{2} , g_{3} , g_{4}) \vphib( g_{1} , g_{2} , g_{3} , g_{5}) \vphi( g_{8} , g_{7} , g_{6} , g_{5}) \\ 
&& \vphib( g_{8} , g_{9} , g_{10} , g_{11}) \vphi( g_{12} , g_{9} , g_{10} , g_{11}) \vphib( g_{12} , g_{7} , g_{6} , g_{4})\,.
\ees

It is an invariant, represented by the $4$-colored graph of Figure \ref{int6}.
It is moreover melonic, and its external legs can be paired so as to form the vacuum melopole shown in Figure \ref{melop6}. This melopole strictly contains four
non-empty melopoles: $S_1 = \{ l_1\}$, $S_3 = \{ l_3\}$, $S_{12} = \{ l_1 , l_2\}$, $S_{23} = \{ l_2 , l_3 \}$. On the other hand, $\{ l_1 , l_3 \}$ and $\{ l_2 \}$ are not melopoles. 

\begin{figure}[h]
  \centering
  \subfloat[Interaction]{\label{int6}\includegraphics[scale=0.5]{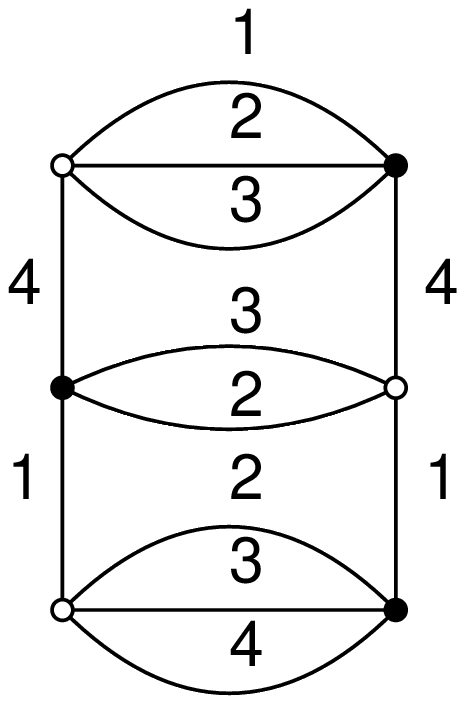}}                
  \subfloat[Vacuum melopole]
{\label{melop6}\includegraphics[scale=0.5]{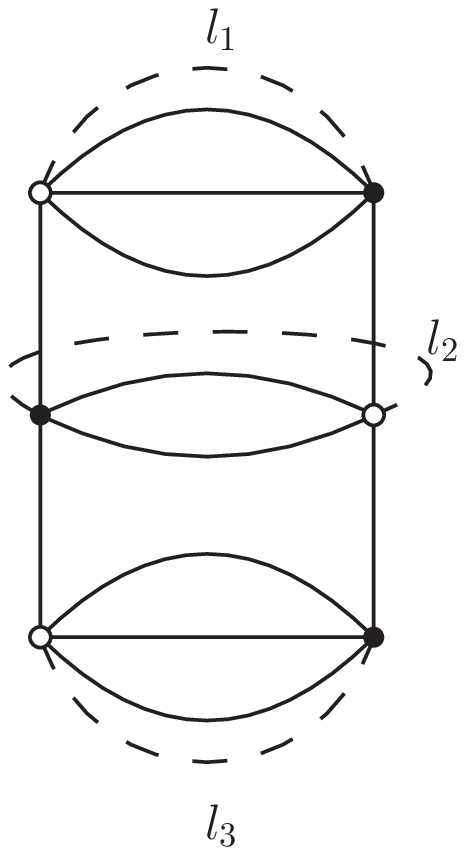}}
  \caption{Bubble interaction with color labels, and unique vacuum melopole that can be obtained from it.}
\end{figure}

We can construct $16$ meloforests out of these melopoles. Half of them, hence 8 do not contain the full graph $S_{123}$.
They are listed below according to the number of subgraphs:
\begin{itemize}
\item the empty forest $\emptyset$;
\item 4 forests with $1$ subgraph: $\{ S_{1} \}$, $\{ S_{3} \}$, $\{ S_{12} \}$, $\{ S_{23} \}$;
\item 3 forests with $2$ subgraphs: $\{ S_{1} , S_{12} \}$, $\{ S_{3} , S_{23} \}$, $\{ S_{1} , S_{3} \}$.
\end{itemize} 
The other half is simply obtained by adding $S_{123}$ to all of these forests. 

\
The melordering generates three kinds of counter-terms: vacuum terms, $2$-point function terms, and two types of $4$-point function terms. We call $b_2$ the $2$-point effective bubble, $b_{4,1}$ and $b_{4,4}$ the two $4$-point effective bubbles, as shown in Figure \ref{eff6}. The melordered interaction will take the form
\beq
\Omega_\rho (S) = S + t_{4,1} (\rho) \, b_{4,1} + t_{4,4} (\rho) \, b_{4,4} + t_{2} (\rho) \, b_{2} + t_\emptyset (\rho) \,,
\eeq
where $t$ are sums of products of coefficients $\nu$. To determine them, we need to analyze the contraction operators they correspond to.

\begin{figure}
\begin{center}
\includegraphics[scale=0.5]{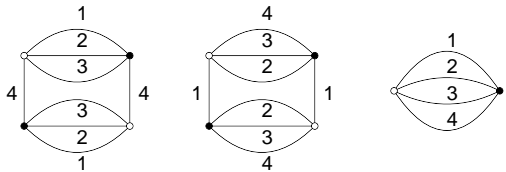}
\caption{Effective interactions generated by melordering. From left to right: $b_{4,4}$, $b_{4,1}$ and $b_2$.}
\label{eff6}
\end{center}
\end{figure}

\
The $4$-point interaction terms are simple, since they are generated by forests $\{ S_1 \}$ and $\{ S_3 \}$. $t_{4,1} (\rho)$ and $t_{4,4} (\rho)$ are therefore both given by the evaluation of $\nu_\rho$ on a single-line melopole (noted $\nu_\rho (1)$) 
\beq
t_{4,1} (\rho) = t_{4,4} (\rho) = - \nu_\rho (1) = - \int_{M^{-2 \rho}}^{+ \infty} \extd \alpha e^{-\alpha m^2} \int \extd \lambda \left( K_\alpha (\lambda) \right)^{3}\,,
\eeq
which is proportional to $\rho$ in the large $\rho$ limit. As expected, these are log-divergent terms.

\
The $2$-point interaction is generated by $\{ S_{12} \}$, $\{ S_{23} \}$, $\{ S_{1} , S_{12} \}$, $\{ S_{3} , S_{23} \}$ and
$\{ S_{1} , S_{3} \}$. $\{ S_{12} \}$ and $\{ S_{23} \}$ contribute with a minus sign, and with an absolute value given by the evaluation of a two-line melopole, that is
\beq
- \nu_\rho (2) = \int_{M^{-2 \rho}}^{+ \infty} \extd \alpha_1 e^{-\alpha_1 m^2} \int_{M^{-2 \rho}}^{+ \infty} \extd \alpha_2 e^{-\alpha_2 m^2}
\int \extd \lambda_1 \int \extd \lambda_2 \left( K_{\alpha_1} (\lambda_1) \right)^{2} \left( K_{\alpha_2} (\lambda_2) \right)^{3} K_{\alpha_1 + \alpha_2} (\lambda_1 + \lambda_2)
\eeq
each. The three other terms come with a plus sign, and factorize as the square of a single line melopole. Therefore:
\beq
t_2 (\rho) = - 2 \nu_\rho (2) + 3 ( \nu_\rho (1) )^2 \,.
\eeq

\
All the other forests contribute to the vacuum counter-term. There are eight of them. It is then easy to see that:
\beq
t_\emptyset (\rho) = - \mu_\rho (3) + 2 \nu_\rho (1) \mu_\rho (2)  + 2 \nu_\rho (2) \mu_\rho (1)  - 3 ( \nu_\rho (1) )^2 \mu_\rho (1)  \,.
\eeq
where the $\nu_\rho$ are the logarithmically divergent previous integrals and the 
$\mu_\rho(1,2,3)$ are full vacuum melopoles amplitudes (with respectively 1 2 and 3 lines), each diverging 
linearly in $M^{ \rho}$. One can check that the integral over the Gaussian measure
of the full melordered combination is then finite as all divergent contributions cancel out.   

%% file: sections/biblio.tex

%% file: pphi.bbl
\begin{thebibliography}{99}

\bibitem{LQG} T. Thiemann, {\it Modern canonical quantum General Relativity}, Cambridge University Press, Cambridge (2007); A. Ashtekar, J. Lewandowski (2004) Background independent quantum
gravity: A status report, {\em Class Quant Grav} {\bf 21} R53-R152; C. Rovelli, {\it Quantum Gravity}, Cambridge University Press, Cambridge (2006)
\bibitem{DT} J. Ambjorn, A. Goerlich, J. Jurkiewicz, R. Loll, arXiv:1203.3591 [hep-th]
\bibitem{qRC} H. Hamber, Gen.Rel.Grav. 41 (2009) 817-876, arXiv:0901.0964 [gr-qc]; R. M. Williams, Nucl. Phys. Proc. Suppl. 57 (1997) 73-81, gr-qc/9702006
\bibitem{ST} N. Seiberg, hep-th/0601234; E. Witten,  hep-th/9306122 [hep-th]; G. Horowitz, J. Polchinski, in {\sl Approaches to Quantum Gravity}, D. Oriti (ed.), 169-186, Cambridge University Press (2009), gr-qc/0602037 [gr-qc]

\bibitem{topology} F. Dowker, R. Sorkin, Class. Quant. Grav. 15, 1153-1167 (1998), gr-qc/9609064; P. Anspinwall, B. Greene, D. Morrison, Nucl. Phys. B 416, 414-480 (1994), hep-th/9309097; T. Banks, Nucl. Phys. B 309, 493 (1988); S. Coleman, Nucl. Phys. B 310, 643 (1988); S. Giddings, A. Strominger, Nucl. Phys. B \textbf{321}, 481, (1989)

 \bibitem{MM} F. David,  Nucl. Phys. B257, \textbf{45} (1985); P. Ginsparg, [arXiv: hep-th/9112013]; P. Di Francesco, P. Ginsparg, J. Zinn-Justin, Phys. Rept. 254 (1995) 1-133,hep-th/9306153; P.
Ginsparg, G. Moore, hep-th/9304011
\bibitem{sorkin} R. Sorkin, Stud.Hist.Philos.Mod.Phys. 36 (2005) 291-301, hep-th/0504037 [hep-th]

\bibitem{SF} A. Perez, Living Reviews, to appear, arXiv:1205.2019
\bibitem{zakopane} C. Rovelli, arXiv:1102.3660



\bibitem{GFT1} D. Oriti, in  {\sl Approaches to Quantum Gravity}, D. Oriti, ed., Cambridge
University Press, Cambridge (2009), [arXiv: gr-qc/0607032]


\bibitem{GFT2} D. Oriti, in {\sl Quantum gravity}, Fauser, B. (ed.) et al., 101-126, Birkhauser (2006), gr-qc/0512103 [gr-qc]

\bibitem{tensorReview} R. Gurau, J. Ryan, arXiv:1109.4812 [hep-th]



\bibitem{GFT3} D. Oriti, 
in  {\sl Foundations of space and time}, G. Ellis, J. Murugan, A. Weltman (eds), Cambridge University Press, Cambridge (2012), arXiv:1110.5606 [hep-th]


\bibitem{vincentTensor} V. Rivasseau, arXiv:1112.5104 [hep-th]



\bibitem{tensor} M. Gross,  Nucl. Phys. Proc. Suppl. \textbf{25A}, 144-149, (1992); J. Ambjorn, B. Durhuus, T. Jonsson,  Mod. Phys. Lett. \textbf{A6}, 1133-1146, (1991); N. Sasakura, Mod.Phys.Lett. A6
(1991) 2613-2624


\bibitem{mikecarlo} M. Reisenberger, C. Rovelli, Class. Quant. Grav. 18 (2001) 121-140, gr-qc/0002095
\bibitem{aristidedaniele} 
  A.~Baratin and D.~Oriti,
  Phys.\ Rev.\ Lett.\  {\bf 105}, 221302 (2010)
  [arXiv:1002.4723 [hep-th]]; A. Baratin, B. Dittrich, D. Oriti, J. Tambornino,s  Class.Quant.Grav. 28 (2011) 175011, arXiv:1004.3450 [hep-th]


\bibitem{boulatov} D. V Boulatov, Mod.Phys.Lett. A{\bf 7}:1629-1646 (1992), [arXiv:hep-th/9202074]

\bibitem{ooguri} H. Ooguri, Mod. Phys. Lett. A7, 2799 (1992), hep-th/9205090




\bibitem{DPFKR} R. De Pietri, L. Freidel, K. Krasnov, C. Rovelli, Nucl. Phys. B \textbf{574}, 785 (2000), [arXiv: hep-th/9907154]; A. Perez, C. Rovelli, Nucl. Phys. B \textbf{599}, 255 (2001), [arXiv: gr-qc/0006107]

\bibitem{EPRL} L. Freidel, K. Krasnov, Class. Quant. Grav. \textbf{25}, 125018 (2008) [arXiv: 0708.1595];  J. Engle, R. Pereira, C. Rovelli, Nucl. Phys. B \textbf{798}, 251 (2008), [arXiv:
0708.1236]; J. Engle, E. Livine, R. Pereira, C. Rovelli, Nucl. Phys. B 
\textbf{799}, 136 (2008), [arXiv:0711.0146];   J. Ben Geloun, R. Gurau, V. Rivasseau, Europhys.Lett. 92 (2010) 60008, arXiv:1008.0354 [hep-th]

\bibitem{BO-BC} A. Baratin, D. Oriti,   New J. Phys. 13 (2011) 125011, arXiv:1108.1178 [gr-qc]
  
\bibitem{BO-Immirzi} A. Baratin, D. Oriti,   Phys. Rev. D85 (2012) 044003, arXiv:1111.5842 [hep-th]  

\bibitem{GFTfluid} D. Oriti, Proceedings of Science PoS(QG-Ph)030, [arXiv:0710.3276]


\bibitem{lorenzoGFT}   L.~Sindoni,
 Gravity as an emergent phenomenon: a GFT perspective,
 arXiv:1105.5687 [gr-qc]

\bibitem{graphity} T. Konopka, F. Markopoulou, L. Smolin, hep-th/0611197;  T. Konopka, F. Markopoulou, S. Severini, Phys. Rev. D77 (2008) 104029, arXiv:0801.0861 [hep-th]

\bibitem{joseph} J. Ben Geloun, J.Math.Phys. 53 (2012) 022901, arXiv:1107.3122 [hep-th]
\bibitem{GFTdiffeos} 
  A.~Baratin, F.~Girelli and D.~Oriti,
  Phys.\ Rev.\ D {\bf 83}, 104051 (2011)
  [arXiv:1101.0590 [hep-th]]


\bibitem{virasoro} R. Gurau, Nucl.Phys. B852 (2011) 592-614, arXiv:1105.6072 [hep-th]



\bibitem{danielelorenzo} D. Oriti, L. Sindoni, New J.Phys. 13
(2011) 025006, arXiv:1010.5149 [gr-qc]


\bibitem{danieleflorianetera} F. Girelli, E. Livine and D. Oriti, Phys. Rev. D81, 024015 (2010), [arXiv: 0903.3475 [gr-qc]]


\bibitem{effHamilt} E. Livine, D. Oriti, J. Ryan, Class.Quant.Grav. 28
(2011) 245010, arXiv:1104.5509 [gr-qc]


\bibitem{gfc} G. Calcagni, S. Gielen, D. Oriti, Class.Quant.Grav. 29 (2012) 105005, arXiv:1201.4151 [gr-qc]





\bibitem{bianca} B. Bahr, B. Dittrich, Phys.Rev. D80 (2009)
124030, arXiv:0907.4323 [gr-qc], B. Dittrich, F. Eckert, M.
Martin-Benito, arXiv:1109.4927 [gr-qc]


\bibitem{GFTrenorm} L. Freidel, R. Gurau and D. Oriti, Phys. Rev. D \textbf{80}, 044007 (2009), [arXiv:0905.3772]; V. Rivasseau, PoS CNCFG2010 (2010) 004, arXiv:1103.1900 [gr-qc]


\bibitem{lin} J. Ben Geloun, T. Krajewski, J. Magnen, V. Rivasseau, Class.Quant.Grav. 27 (2010) 155012, arXiv:1002.3592 [hep-th]

\bibitem{valentinmatteo} V. Bonzom, M. Smerlak, Lett.Math.Phys. 93 (2010) 295-305, arXiv:1004.5196 [gr-qc]; V. Bonzom, M. Smerlak, arXiv:1008.1476 [math-ph], V. Bonzom, M. Smerlak, arXiv:1103.3961
[gr-qc]


\bibitem{scaling}
  S.~Carrozza and D.~Oriti,
  Phys.\ Rev.\  D {\bf 85}, 044004 (2012)
  [arXiv:1104.5158 [hep-th]]; S. Carrozza, D. Oriti, JHEP 1206 (2012) 092, arXiv:1203.5082 [hep-th]

\bibitem{jimmy} J. Ryan, Phys.Rev. D85 (2012) 024010, arXiv:1104.5471 [gr-qc]

\bibitem{francesco} F. Caravelli, SpringerPlus 20 1 (1:6),  arXiv:1012.4087 [math-ph]

\bibitem{divergences4dGFT} L. Crane, A. Perez, C. Rovelli, Phys.Rev.Lett. 87 (2001) 181301, gr-qc/0104057 [gr-qc];  
A. Perez, Nucl.Phys. B599 (2001) 427-434, gr-qc/0011058 [gr-qc]; C. Perini, C. Rovelli, S. Speziale, Phys.Lett. B682 (2009) 78-84, arXiv:0810.1714 [gr-qc]


\bibitem{Gurau:2009tw} 
  R.~Gurau,
  Commun.\ Math.\ Phys.\  {\bf 304}, 69 (2011),
 arXiv:0907.2582 [hep-th]


\bibitem{crystallization} M. Ferri and C. Gagliardi, {\sl Crystallisation moves.},
Pacific J. Math. Volume 100, Number 1 (1982), 85-103, A. Vince, Discrete Mathematics, Volume 72, Issues 1–3, 367-380 (1988); A. Vince, Graphs and Combinatorics, 9:75 84 (1993)
  
  
\bibitem{uncoloring}
V. Bonzom, R. Gurau, V. Rivasseau, arXiv:1202.3637 [hep-th]

  
\bibitem{universality} R. Gurau, arXiv:1111.0519 [math.PR]


\bibitem{large-N}
  R.~Gurau,
  Annales Henri Poincare {\bf 12}, 829 (2011)
  [arXiv:1011.2726 [gr-qc]];  
  R.~Gurau and V.~Rivasseau,
  Europhys.\ Lett.\  {\bf 95}, 50004 (2011)
  [arXiv:1101.4182 [gr-qc]];
  R.~Gurau,
  arXiv:1102.5759 [gr-qc]
  
  

\bibitem{critical} V. Bonzom, R. Gurau, A. Riello, V. Rivasseau, Nucl.Phys. B853 (2011) 174-195, arXiv:1105.3122 [hep-th]





\bibitem{critical2} V. Bonzom, R. Gurau, V. Rivasseau, arXiv:1108.6269 [hep-th]; 
D. Benedetti, R. Gurau, Nucl.Phys. B855 (2012) 420-437, arXiv:1108.5389 [hep-th]



\bibitem{generalisedGFT} D. Oriti, Phys. Rev. D73 (2006) 061502, gr-qc/0512069; D. Oriti,  Class. Quant. Grav. 27 (2010) 145017, arXiv:0902.3903 [gr-qc]; D. Oriti, T. Tlas, Class. Quant. Grav. 27
(2010) 135018, arXiv:0912.1546 [gr-qc]

\bibitem{josephvalentin} 
  J.~Ben Geloun and V.~Bonzom,
  Int.\ J.\ Theor.\ Phys.\  {\bf 50}, 2819 (2011)
  [arXiv:1101.4294 [hep-th]]

\bibitem{tensor_4d} J. Ben Geloun, V. Rivasseau,   arXiv:1111.4997 [hep-th]

\bibitem{josephsamary} J. Ben Geloun, D. Ousmane Samary,   arXiv:1201.0176 [hep-th]

\bibitem{gelounlivine} J. Ben Geloun and E. Livine,  arXiv:1207.0416 [hep-th]


\bibitem{josephaf}
  J.~Ben Geloun,
  arXiv:1205.5513 [hep-th]

  
  
\bibitem{Rivasseau:2007fr} 
  V.~Rivasseau,
  JHEP {\bf 0709}, 008 (2007)
  [arXiv:0706.1224 [hep-th]]; 
  J.~Magnen and V.~Rivasseau,
  Annales Henri Poincare {\bf 9}, 403 (2008), arXiv:0706.2457 [math-ph]
  
\bibitem{Magnen:2009at} 
  J.~Magnen, K.~Noui, V.~Rivasseau and M.~Smerlak,
  Class.\ Quant.\ Grav.\  {\bf 26}, 185012 (2009), arXiv:0906.5477 [hep-th]; R. Gurau, arXiv:1111.0519 [math.PR]
  
\bibitem{Rivasseau:2010ke}
  V.~Rivasseau and Z.~Wang,
  J.\ Math.\ Phys.\  {\bf 51} (2010) 092304, arXiv:1003.1037 [math-ph]


 \bibitem{bonzomerbin} 
 V.~Bonzom and H. Erbin, arXiv:1204.3798 hep-th]
 
  \bibitem{BGS} 
  V.~Bonzom, R.~Gurau and M. Smerlak, arXiv:1206.5539 [hep-th]
  
\bibitem{Simon} B. Simon, 
$P(\Phi)_2$ Euclidean Quantum Field Theory, Princeton University Press, 1974

\bibitem{VincentBook} V. Rivasseau, {\sl From perturbative to constructive renormalization}, Princeton University Press (1991)

\bibitem{GW} H. Grosse, R. Wulkenhaar, Commun. Math. Phys. 256 (2005) 305-374, hep-th/0401128; V. Rivasseau, F. Vignes-Tourneret, R. Wulkenhaar,  Commun. Math. Phys. 262 (2006) 565-594,
arXiv:0501036 [hep-th]

\bibitem{FLouapre} 
L. Freidel and David Louapre,
Class.Quant.Grav. 21 (2004) 5685-5726, arXiv:0401076 [hep-th]

\bibitem{Rivasseau:2011df} 
  V.~Rivasseau and Zhituo~Wang,
  J.\ Math.\ Phys.\  {\bf 53}, 042302 (2012), arXiv:1104.3443 [math-ph]
  
\bibitem{ZW2} Zhituo Wang, arXiv:1205.0196 [hep-th]


\end{thebibliography}
